\documentclass[journal]{IEEEtran}
\usepackage{cite}
\usepackage{amsmath}
\usepackage{amssymb}
\usepackage{amsfonts}
\usepackage{cases}
\usepackage{graphicx}
\usepackage{color}
\usepackage{multirow}
\usepackage{booktabs}
\usepackage{relsize}
\usepackage{stfloats}
\usepackage{xcolor}
\usepackage{enumerate}
\usepackage{float}
\usepackage{setspace}
\usepackage{verbatim}
\usepackage{bm}
\usepackage{xfrac}
\usepackage{subcaption}
\usepackage{dsfont}
\usepackage{lipsum}
\usepackage{hyperref}
\usepackage{algorithm}
\usepackage{algpseudocode}
\usepackage{pgfplots}
\pgfplotsset{compat=1.18}
\usepgfplotslibrary{polar}
\usepackage{subcaption}

\newtheorem{remark}{Remark}

\newtheorem{proposition}{Proposition}

\newcommand{\rom}[1]{\uppercase\expandafter{\romannumeral #1\relax}}

\allowdisplaybreaks

\begin{document}

\title{Electromagnetic Modeling of Deformation-Based Pinching Antennas}

\author{Dimitrios~Tyrovolas,~\IEEEmembership{Member,~IEEE}, Nikos C. Sagias,~\IEEEmembership{Senior Member,~IEEE}, Ke Guan,~\IEEEmembership{Senior Member,~IEEE}, \\ George K.~Karagiannidis,~\IEEEmembership{Fellow,~IEEE}, and Stylianos D. Asimonis,~\IEEEmembership{Senior Member,~IEEE}% <-this % stops a space

\thanks{ 
D. Tyrovolas and N. C. Sagias are with the Department of Informatics and Telecommunications, University of Peloponnese, 22131 Tripoli, Greece (e-mail: d.tyrovolas@uop.gr, nsagias@uop.gr).

K. Guan is with the School of Electronics and Information Engineering, Beijing Jiaotong University, 100044, Beijing, China and Christian Doppler Laboratory for Digital Twin Assisted AI for Sustainable Radio Access Networks and Institute for Telecommunications, Technische Universität Wien, 1040 Vienna, Austria (e-mail: kguan@bjtu.edu.cn). 

G. K. Karagiannidis is with the Department of Electrical and Computer Engineering, Aristotle University of Thessaloniki, 54124 Thessaloniki, Greece (e-mail: geokarag@auth.gr).

S. D. Asimonis is with the Department of Electrical and Computer Engineering, University of Patras, 26504 Patras, Greece (e-mail: s.asimonis@upatras.gr)

}
}

\maketitle

%==============================================================
\begin{abstract}
Pinching-antenna systems (PAS) rely on local perturbations of a guiding structure to radiate part of the guided electromagnetic energy, making the physical response of the resulting pinching antenna (PA) fundamental to both guided-wave propagation and free-space radiation. To capture this coupled behavior, we develop an electromagnetic model of a finite deformation-based PA that directly relates the applied waveguide perturbation to its complex reflection and transmission responses and three-dimensional radiation characteristics. Full-wave MEEP simulations validate the model and demonstrate how the pinching ratio and pinched-section length jointly affect the transmitted power and phase, radiation power-fraction, and radiation pattern, thereby establishing PA geometry as a degree of freedom for controlling the radiation response. Importantly, this geometric dependence extends to communication performance, as the received signal-to-noise ratio (SNR) results reveal that neither minimizing the PA--user distance nor maximizing the radiation power-fraction necessarily yields the highest SNR, since the applied perturbation also determines the realized gain toward the user.
\end{abstract}


%==============================================================
\section{Introduction}

In next-generation wireless networks, emerging applications such as immersive extended reality and the metaverse impose stringent requirements on reliability, latency, and data rate, while simultaneously demanding that the wireless infrastructure adapts to diverse user locations and propagation conditions \cite{6GMetaverse}. To meet these requirements, flexible-antenna architectures have been introduced as a means of dynamically modifying the antenna configuration and, consequently, the wireless channel \cite{survey2026}. In more detail, these architectures exploit spatially reconfigurable radiating elements whose locations can be adapted according to the communication objectives. However, realizing large-scale spatial reconfigurability through conventional antenna architectures introduces critical challenges in terms of feeding complexity, deployment cost, and propagation losses, particularly at high frequencies \cite{survey2026}. In such settings, transporting the signal through low-loss guiding structures before radiating it becomes particularly attractive, leading to the need for antenna architectures that can jointly exploit guided-wave propagation and dynamically configurable radiation points \cite{DOCOMO2022}. Therefore, enabling flexible radiation locations while preserving an efficient signal distribution mechanism becomes important for realizing spatially adaptive wireless networks.

Among the technologies envisioned to support such flexible antenna deployments, pinching-antenna systems (PAS) have emerged as a promising approach for providing reconfigurable wireless connectivity \cite{DOCOMO2022,Ding2025PAS,XuGeneralizedPAS2026}. Specifically, in PAS, the radio-frequency signal propagates along a dielectric waveguide and is radiated into free space through PAs that can be activated at different positions along it, thereby enabling the radiation points to be established close to the served users and facilitating favorable line-of-sight links \cite{Ding2025PAS,XuGeneralizedPAS2026}. At the same time, the selection of the PA locations directly affects both the free-space propagation distances and the phase shifts of the corresponding communication links, rendering the achievable performance inherently dependent on the spatial configuration of the activated PAs \cite{XuRatePAS2025,TyrovolasPerformancePAS}. However, the signal radiated by each PA is ultimately generated through its physical interaction with the guided electromagnetic field, which determines how the incident guided power is redistributed between the signal remaining in the waveguide and the field radiated into the surrounding environment \cite{Yuanwei2025Modeling, LiGeometryPAS2026}. As a result, the communication performance of PAS does not only depend on where the PAs are activated, but also on how their physical response modifies the guided and radiated signals. Consequently, understanding the extent to which the physical PA response affects the achievable communication performance becomes critical, motivating the need to systematically characterize PAS under physically consistent models.

\subsection{State-of-the-Art}

An increasing amount of literature has focused on the communication-level performance and design of PAS. In this direction, \cite{TegosPAS2025} investigated PA-position and resource-allocation optimization, while \cite{OikonomouOFDMA2025} characterized the frequency-selective behavior of multi-PA transmission and introduced an OFDMA-based framework. Furthermore, \cite{OuyangArrayGain2025} analyzed the achievable array gain and the impact of the number and spacing of the activated PAs, whereas \cite{TyrovolasPerformancePAS} characterized the outage probability, average rate, and optimal PA location. More recently, PAS deployments have been extended beyond conventional horizontal straight waveguides, with \cite{WangVerticalPAS2026} considering vertical PAS for UAV communications and \cite{QuCurvedPAS2026} investigating curved waveguide deployments. Therefore, these studies show that the achievable PAS performance is strongly affected by the PA configuration and waveguide deployment, motivating an accurate description of the physical mechanisms affecting PAS performance.

The design and performance implications of PAS have also been examined from a hardware and power-coupling perspective, providing further insight into the physical interaction between the guiding structure and the activated PAs. In more detail, \cite{Yuanwei2025Modeling} developed a physics-based hardware model in which the PA is represented as an open-ended directional coupler and coupled-mode theory is employed to characterize the corresponding power transfer, while \cite{XuPowerRadiation2026} proposed an adjustable power-radiation model in which the radiated power is controlled through the coupling spacing between the PA and the waveguide. More recently, \cite{WangMultiportPAS2026} established a physically consistent multiport-network representation for reconfigurable PAs, providing a network-level description of their guided and radiated signal interactions. In parallel, the radiation characteristics of PAs have been investigated through electromagnetic analysis, full-wave simulations, measurements, and measured-pattern-based system models. Specifically, \cite{LiGeometryPAS2026} demonstrated that the PA geometry substantially affects the radiation, \cite{ZubairRadiationPAS2026} developed a radiation-pattern model based on a two-dimensional dielectric slab, coupled-mode theory, and a two-dimensional radiation integral, and \cite{WeiMMPAS2026} introduced a polarization-aware full-wave electromagnetic model for multi-mode PAS. Finally, \cite{FengMeasuredPAS2026} incorporated an externally obtained measured radiation pattern together with coupling-dependent power extraction for PA placement and multi-PA design. Taken together, these studies have progressively established the importance of jointly accounting for PA power coupling, network response, and directional radiation when relating the physical PA behavior to communication-level performance.

\subsection{Motivation \& Contribution}

The practical realization of a PA requires a local perturbation of the guiding structure to couple part of the guided electromagnetic energy into free space \cite{Yuanwei2025Modeling,XuPowerRadiation2026}. In particular, depending on the applied perturbation, it can modify the local geometric characteristics of the waveguide and, consequently, affect both the guided-wave propagation and the corresponding radiation response \cite{LiGeometryPAS2026}. As a result, the PA geometry can jointly determine the radiated power, the amplitude and phase of the signal remaining in the waveguide, and the spatial radiation characteristics, thereby directly affecting the complex signal transmitted along the downstream waveguide and the resulting communication performance. However, existing studies have mainly characterized these effects through power-coupling or radiation-oriented formulations \cite{Yuanwei2025Modeling,XuPowerRadiation2026,LiGeometryPAS2026,ZubairRadiationPAS2026,WeiMMPAS2026}. Therefore, to the best of the authors' knowledge, no existing work has jointly related the same finite physical PA perturbation to its complete complex guided-wave scattering response and corresponding three-dimensional radiation characteristics within a unified electromagnetic framework.

In this paper, we develop an electromagnetic model of a finite deformation-based PA that directly relates the applied waveguide perturbation to its guided-wave and radiation responses. Specifically, starting from the same finite PA geometry, the proposed model jointly derives the complex reflection and transmission coefficients and the corresponding
three-dimensional radiation characteristics, thereby
determining both the complex signal remaining in the
waveguide and the field radiated into free space. Furthermore, validation against full-wave electromagnetic simulations performed using MEEP shows that the pinching ratio and pinched-section length jointly affect the transmitted power and phase, radiation power-fraction, and radiation pattern, thus establishing these geometric parameters as design degrees of freedom for shaping the PA radiation response. Finally, the proposed model is evaluated in terms of the received SNR, showing that minimizing the PA--user distance or maximizing the total radiation power-fraction does not necessarily yield the highest received SNR, since the applied perturbation also determines the realized gain and, consequently, the effectiveness of a given PA position. Therefore, the proposed electromagnetic model directly connects the finite deformation-based PA geometry to its complex downstream guided-wave response and three-dimensional radiation characteristics, providing the physical quantities required for geometry-aware PAS design.

The remainder of the paper is organized as follows. Section~II develops the guided-wave propagation and PA scattering model, while Section~III derives the corresponding far-field radiation characteristics and realized gain. Section~IV presents the numerical results, full-wave validation, and received-SNR evaluation, whereas Section~V concludes the paper.

%==============================================================

\section{Waveguide Propagation and Pinching-Antenna Response}
\label{sec:pa-characterization}

In this section, the propagation of the guided signal along the dielectric waveguide and its modification by an activated PA are characterized. The fundamental HE$_{11}$ mode is first described in the unperturbed and pinched waveguide sections, yielding the radius-dependent propagation constants and normalized modal fields, which are subsequently used to characterize the finite PA response arising from the two radius transitions and the propagation between them.

\subsection{Propagation in a Cylindrical Dielectric Waveguide with Pinching Antennas}
\label{sec:modes}

In a PAS, the guided signal propagates along a dielectric waveguide and is locally modified by the activated PAs. To isolate and quantify the electromagnetic effect of the mechanical perturbation introduced by pinching, an axisymmetric reduced-radius representation is adopted, in which the $k$-th PA locally changes the waveguide radius from $a$ to $b_k=\kappa_k a$ over a finite length $L_{{\rm P},k}$. This representation captures the transverse magnitude and longitudinal extent of the considered pinching perturbation through $\kappa_k$ and $L_{{\rm P},k}$, respectively, while retaining the two radius transitions and the propagation between them that jointly determine the resulting guided-wave scattering and radiation responses~\cite{Eleftheriades1994,Svendsen2013}. The resulting structure therefore consists of a transition from $a$ to $b_k$, a uniform pinched section, and the reverse transition from $b_k$ to $a$, which are modeled as piecewise-uniform cylindrical dielectric waveguide sections following the modal-scattering representation of waveguide discontinuities.

Under the considered single-mode regime, both the unperturbed and
pinched sections support the fundamental hybrid HE$_{11}$
mode~\cite{Snitzer1961}. Accordingly, for a cylindrical dielectric
waveguide section of radius $r\in\{a,b_k\}$, the dependence of this mode
on the local waveguide geometry is introduced through the normalized
frequency
\begin{equation}
\small
V(r)
=
k_0r\sqrt{n_1^2-n_2^2}
=
k_0r\sqrt{\epsilon_r-1},
\label{eq:V}
\end{equation}
where $k_0=2\pi/\lambda_0$ is the free-space wavenumber, while
$n_1=\sqrt{\epsilon_r}$ and $n_2=1$ denote the refractive indices of
the dielectric core and the surrounding air, respectively. For a given
$V(r)$, the transverse behavior of the HE$_{11}$ mode is characterized
by the normalized transverse wavenumbers
$u(r)=r\sqrt{k_0^2\epsilon_r-\beta^2(r)}$ inside the dielectric core
and $w(r)=r\sqrt{\beta^2(r)-k_0^2}$ in the surrounding air.
Consequently, it holds that $u^2(r)+w^2(r)=V^2(r)$ and, therefore,
$w(r)=\sqrt{V^2(r)-u^2(r)}$, while the corresponding propagation
constant follows as
\begin{equation}
\small
\beta(r)
=
\sqrt{
k_0^2+\frac{w^2(r)}{r^2}
}.
\label{eq:beta-from-w}
\end{equation}
Thus, for a specified radius $r$, both $w(r)$ and $\beta(r)$ are
determined once $u(r)$ is known, reducing the modal characterization of
each uniform waveguide section to a single unknown. In particular, this
unknown is determined by enforcing the dielectric-air boundary
conditions at $\rho=r$ for the HE$_{11}$ mode, where continuity of the
tangential electric and magnetic fields yields~\cite{Snitzer1961}
\begin{equation}
\small
\begin{split}
&
\Bigg[
\frac{J'_1\!\left(u(r)\right)}
{u(r)J_1\!\left(u(r)\right)}
+
\frac{K'_1\!\left(w(r)\right)}
{w(r)K_1\!\left(w(r)\right)}
\Bigg]
\Bigg[
n_1^2
\frac{J'_1\!\left(u(r)\right)}
{u(r)J_1\!\left(u(r)\right)}
\\
&\qquad
+
n_2^2
\frac{K'_1\!\left(w(r)\right)}
{w(r)K_1\!\left(w(r)\right)}
\Bigg]
=
\left(
\frac{\beta(r)}{k_0}
\right)^2
\left(
\frac{1}{u^2(r)}
+
\frac{1}{w^2(r)}
\right)^2,
\end{split}
\label{eq:he-general}
\end{equation}
where $J_1(\cdot)$ is the Bessel function of the first kind,
$K_1(\cdot)$ is the modified Bessel function of the second kind, and
$(\cdot)'$ denotes differentiation with respect to the argument.
Consequently, substituting $w(r)=\sqrt{V^2(r)-u^2(r)}$ and \eqref{eq:beta-from-w} into \eqref{eq:he-general} reduces the boundary-condition relation to a one-dimensional nonlinear equation for $u(r)\in(0,V(r))$, from which the root corresponding to the fundamental HE$_{11}$ branch is retained to determine $w(r)$ and $\beta(r)$. Finally, evaluating this
modal solution for $r=a$ and $r=b_k$ yields the propagation constants
$\beta(a)$ and $\beta(b_k)$ of the unperturbed waveguide and the
pinched section, respectively, thereby providing the guided phase
evolution along the common waveguide and through the finite pinched
section, together with the modal quantities required to characterize
the scattering generated at the two radius transitions.

At each radius transition, the local guided response is determined by
the HE$_{11}$ field distributions supported by the adjoining
unperturbed and pinched sections. In particular, the change in radius
between $a$ and $b_k$ modifies the transverse confinement of the guided
mode and therefore produces different field distributions on the two
sides of the transition. Accordingly, for the modal characterization, a
local cylindrical coordinate system $(\rho,\phi,z)$ is adopted with the
waveguide axis along $z$, while $\rho$ and $\phi$ denote the radial and
azimuthal coordinates in the transverse plane, respectively. For a
uniform cylindrical section of radius $r\in\{a,b_k\}$, the longitudinal
electric and magnetic components of the forward-propagating HE$_{11}$
mode are expressed as~\cite{Snitzer1961}
{\small
\begin{align}
E_z(\rho,\phi,z;r)
&=
C(r)
\Psi(\rho;r)
\cos\phi\,
e^{-j\beta(r)z},
\label{eq:Ez-he11}
\\
H_z(\rho,\phi,z;r)
&=
D(r)
\Psi(\rho;r)
\sin\phi\,
e^{-j\beta(r)z},
\label{eq:Hz-he11}
\end{align}
}
where the corresponding radial field dependence is
\begin{equation}
\small
\Psi(\rho;r)
=
\begin{cases}
J_1\!\left(
u(r)\rho/r
\right),
&
0\leq\rho<r,
\\[2.2ex]
\dfrac{
J_1\!\left(u(r)\right)
}{
K_1\!\left(w(r)\right)
}
K_1\!\left(
w(r)\rho/r
\right),
&
\rho\geq r.
\end{cases}
\label{eq:psi-he11}
\end{equation}
Since the HE$_{11}$ mode is hybrid, both longitudinal field components
are required to satisfy the dielectric-air boundary conditions, and their
amplitudes are related according to~\cite{Snitzer1961}
\begin{equation}
\small
\frac{D(r)}{C(r)}
=
-
\frac{\beta(r)}{\omega\mu_0}
\frac{
\displaystyle
\frac{1}{u^2(r)}
+
\frac{1}{w^2(r)}
}{
\displaystyle
\frac{
J_1'\!\left(u(r)\right)
}{
u(r)J_1\!\left(u(r)\right)
}
+
\frac{
K_1'\!\left(w(r)\right)
}{
w(r)K_1\!\left(w(r)\right)
}
},
\label{eq:HE11-longitudinal-ratio}
\end{equation}
where $\omega=2\pi f_0$ and $\mu_0$ denotes the free-space permeability. Accordingly, once the longitudinal fields in \eqref{eq:Ez-he11} and
\eqref{eq:Hz-he11} are specified, the corresponding transverse fields
follow directly from Maxwell's equations. In particular, the transverse
electric field is given by~\cite{Snitzer1961}
\begin{equation}
\small
\begin{split}
&\mathbf{E}_{t}(\rho,\phi,z;r)
=
-\frac{j e^{-j\beta(r)z}}
{k_t^2(\rho;r)}
\Bigg\{
\widehat{\boldsymbol{\rho}}\cos\phi
\left[
\beta(r)C(r)
\frac{{\rm d}\Psi(\rho;r)}{{\rm d}\rho}
\right.
\\
&\left.
+
\omega\mu_0D(r)
\frac{\Psi(\rho;r)}{\rho}
\right]
-
\widehat{\boldsymbol{\phi}}\sin\phi
\left[
\beta(r)C(r)
\frac{\Psi(\rho;r)}{\rho}
\right.
\\
&\left.
+
\omega\mu_0D(r)
\frac{{\rm d}\Psi(\rho;r)}{{\rm d}\rho}
\right]
\Bigg\},
\end{split}
\label{eq:Et-he11}
\end{equation}
while the transverse magnetic field follows as
\begin{equation}
\small
\begin{split}
&\mathbf{H}_{t}(\rho,\phi,z;r)
=
-\frac{j e^{-j\beta(r)z}}
{k_t^2(\rho;r)}
\Bigg\{
\widehat{\boldsymbol{\rho}}\sin\phi
\left[
\beta(r)D(r)
\frac{{\rm d}\Psi(\rho;r)}{{\rm d}\rho}
\right.
\\
&\left.
+
\omega\varepsilon(\rho;r)C(r)
\frac{\Psi(\rho;r)}{\rho}
\right]
+
\widehat{\boldsymbol{\phi}}\cos\phi
\left[
\beta(r)D(r)
\frac{\Psi(\rho;r)}{\rho}
\right.
\\
&\left.
+
\omega\varepsilon(\rho;r)C(r)
\frac{{\rm d}\Psi(\rho;r)}{{\rm d}\rho}
\right]
\Bigg\}.
\end{split}
\label{eq:Ht-he11}
\end{equation}
Moreover, the radial derivative required in
\eqref{eq:Et-he11} and \eqref{eq:Ht-he11} follows directly from
\eqref{eq:psi-he11} as
\begin{equation}
\small
\frac{{\rm d}\Psi(\rho;r)}{{\rm d}\rho}
=
\begin{cases}
\frac{u(r)}{r}
J_0\!\left(
\frac{u(r)}{r}\rho
\right)
-
\frac{1}{\rho}
J_1\!\left(
\frac{u(r)}{r}\rho
\right),
&
0\leq\rho<r,
\\[1mm]
\begin{aligned}
&-
\frac{
J_1\!\left(u(r)\right)
}{
K_1\!\left(w(r)\right)
}
\Bigg[
\frac{w(r)}{r}
K_0\!\left(
\frac{w(r)}{r}\rho
\right)
\\
&\hspace{18mm}
+
\frac{1}{\rho}
K_1\!\left(
\frac{w(r)}{r}\rho
\right)
\Bigg]
\end{aligned},
&
\rho\geq r,
\end{cases}
\label{eq:psi-he11-derivative}
\end{equation}
while $\varepsilon(\rho;r)$ is equal to
\begin{equation}
\small
\varepsilon(\rho;r)
=
\begin{cases}
\epsilon_0\epsilon_r,
&
0\leq\rho<r,
\\
\epsilon_0,
&
\rho\geq r,
\end{cases}
\label{eq:section-permittivity}
\end{equation}
which determines the corresponding transverse wavenumber as
\begin{equation}
\small
\begin{split}
k_t^2(\rho;r)
&=
\omega^2\mu_0\varepsilon(\rho;r)
-
\beta^2(r)
=
\begin{cases}
u^2(r)/r^2,
&
0\leq\rho<r,
\\[1ex]
-w^2(r)/r^2,
&
\rho\geq r.
\end{cases}
\end{split}
\label{eq:section-transverse-wavenumber}
\end{equation}
Hence, evaluating
\eqref{eq:Ez-he11}--\eqref{eq:section-transverse-wavenumber}
for $r=a$ and $r=b_k$ provides the HE$_{11}$ modal fields supported by
the unperturbed and pinched sections, respectively, while the corresponding backward mode can be defined by
$\beta^{-}(r)=-\beta^{+}(r)$, $C^{-}(r)=-C^{+}(r)$, and
$D^{-}(r)=D^{+}(r)$, consistently with the usual forward--backward
waveguide-mode convention~\cite{Svendsen2013}. Hence, at the reference
plane, the transverse fields satisfy
$\mathbf E_t^{-}=\mathbf E_t^{+}$ and
$\mathbf H_t^{-}=-\mathbf H_t^{+}$, while $u(r)$, $w(r)$, and
$\Psi(\rho;r)$ remain unchanged since they depend on $\beta^2(r)$.

With the radius-dependent HE$_{11}$ field distributions established,
their amplitudes must be referred to the same guided-power level before
the interaction at a radius transition can be characterized. Although
\eqref{eq:HE11-longitudinal-ratio} fixes the ratio between $C(r)$ and
$D(r)$, their common modal scaling remains arbitrary, so equal field
coefficients for $r=a$ and $r=b_k$ do not generally correspond to equal
guided powers. Therefore, the magnitude of the power carried by the
HE$_{11}$ mode in a section of radius $r$ is obtained from the
longitudinal component of the time-average Poynting vector
as~\cite{Snitzer1961,Balanis2016}
\begin{equation}
\small
P_{\rm g}(r)=\frac{1}{2}\left|
\Re\!\left\{
\int_{0}^{2\pi}\!\!\!\!\int_{0}^{\infty}\!\!\!
\left[\mathbf E_t(\rho,\phi,z;r)\!\times\!\mathbf H_t^{*}(\rho,\phi,z;r)\right]
\!\cdot\!\widehat{\mathbf z}\,\rho\,{\rm d}\rho\,{\rm d}\phi\!
\right\}\!\right|\!.
\label{eq:guided-modal-power}
\end{equation}
Consequently, the corresponding unit-power modal fields are
\begin{align}
\small
\overline{\mathbf E}(\rho,\phi,z;r)
&=
\frac{
\mathbf E(\rho,\phi,z;r)
}{
\sqrt{P_{\rm g}(r)}
},
\label{eq:power-normalized-electric-field}
% \end{equation}
% and
% \begin{equation}
\\
\small
\overline{\mathbf H}(\rho,\phi,z;r)
&=
\frac{
\mathbf H(\rho,\phi,z;r)
}{
\sqrt{P_{\rm g}(r)}
},
\label{eq:power-normalized-magnetic-field}
\end{align}
where $\mathbf E=\widehat{\mathbf z}E_z+\mathbf E_t$ and
$\mathbf H=\widehat{\mathbf z}H_z+\mathbf H_t$, while
$\overline{\mathbf E}_t$ and $\overline{\mathbf H}_t$ denote the
corresponding transverse components of the normalized fields. Finally, since the same normalization is applied to the forward and backward HE$_{11}$ modes, it also holds that
$\overline{\mathbf E}_t^{-}=\overline{\mathbf E}_t^{+}$ and
$\overline{\mathbf H}_t^{-}=-\overline{\mathbf H}_t^{+}$ at the
reference plane, hence, placing the HE$_{11}$ fields of the unperturbed and pinched sections on the same unit-guided-power basis.

\subsection{Pinching Antenna Response}
\label{sec:mm}
In a PAS, activating a PA locally modifies the
waveguide geometry and redistributes the incident guided power among the signal transmitted to the downstream waveguide section, the signal reflected toward the source, and radiation. Hence, the PA cannot be characterized
through an independently assigned radiation coefficient, but through the
scattering response induced by its finite geometry. In particular, the
$k$-th PA consists of a first transition from the unperturbed radius
$a$ to the pinched radius $b_k$, a uniform section of radius $b_k$ and
length $L_{{\rm P},k}$, and a second transition from $b_k$ back to $a$.
Within the considered single-mode operating regime, the pinched section
supports the fundamental HE$_{11}$ mode, with propagation constant
$\beta(b_k)$ obtained from \eqref{eq:beta-from-w}. Therefore, the PA
response is determined by the field interactions at the two radius
transitions together with the propagation through the finite section
between them.

At each radius transition, the change in waveguide radius modifies the
transverse field distribution of the supported HE$_{11}$ mode.
Consequently, the local guided interaction between the adjoining
waveguide sections is characterized by projecting the interface fields
onto their supported modes. Following the complex-power mode-matching
formulation for waveguide step discontinuities
\cite{Eleftheriades1994}, and retaining the HE$_{11}$ contribution in
the two adjoining sections, the corresponding electric--magnetic
cross-overlaps are defined as
{\small
\begin{align}
W_k
&=
\frac{1}{2}
\int_{0}^{2\pi}
\int_{0}^{\infty}
\left[
\overline{\mathbf E}_{t}(\rho,\phi,0;a)
\times
\overline{\mathbf H}_{t}^{*}(\rho,\phi,0;b_k)
\right]
\cdot
\widehat{\mathbf z}\,
\rho\,{\rm d}\rho\,{\rm d}\phi ,
\label{eq:W-step}
\\
Q_k
&=
\frac{1}{2}
\int_{0}^{2\pi}
\int_{0}^{\infty}
\left[
\overline{\mathbf E}_{t}^{*}(\rho,\phi,0;b_k)
\times
\overline{\mathbf H}_{t}(\rho,\phi,0;a)
\right]
\cdot
\widehat{\mathbf z}\,
\rho\,{\rm d}\rho\,{\rm d}\phi ,
\label{eq:Q-step}
\end{align}
}where $W_k$ and $Q_k$ measure the HE$_{11}$ cross-overlaps obtained by
interchanging the electric and magnetic fields of the unperturbed and
pinched sections. Since the modal fields are normalized to unit guided
power, the corresponding HE$_{11}$ self-overlaps are equal to one. For
unit HE$_{11}$ incidence from the unperturbed section, let $a_{r,k}$
and $a_{t,k}$ denote the local reflected and transmitted HE$_{11}$
amplitudes, respectively. Thus, by projecting the tangential-field continuity conditions onto the normalized HE$_{11}$ mode of the pinched section \cite{Svendsen2013}, and using the forward--backward field convention, gives $\left(1+a_{r,k}\right)W_k=a_{t,k}$ and
$\left(1-a_{r,k}\right)Q_k=a_{t,k}$, from which we can express $a_{r,k}$ and $a_{t,k}$ as
{\small
\begin{align}
a_{r,k}
&=
\frac{Q_k-W_k}{Q_k+W_k},
\label{eq:ar-step}
\\
a_{t,k}
&=
\frac{2W_kQ_k}{Q_k+W_k}.
\label{eq:at-step}
\end{align}
}Therefore, $a_{t,k}$ and $a_{r,k}$ characterize the transmitted and reflected HE$_{11}$ components of the local step response. Finally, with the HE$_{11}$ modal phases chosen consistently at the transition reference plane, reversing the transition from $a$ to $b_k$ into the transition from $b_k$ to $a$ interchanges the two cross-overlap relations. Consequently, under the same modal projection, the two radius transitions have the same local transmission amplitude $a_{t,k}$ and opposite reflection amplitudes $a_{r,k}$ and $-a_{r,k}$, respectively \cite{Svendsen2013}.

The local step response obtained above characterizes the guided
HE$_{11}$ field transmitted and reflected at each radius transition.
However, a dielectric-waveguide discontinuity also excites radiative
and evanescent fields~\cite{Bienstman2001,Khalil2003}. Since the two
radius transitions are separated by the finite distance
$L_{{\rm P},k}$, these additional fields can propagate or decay from
one transition to the other and couple back into the guided HE$_{11}$
field, thereby introducing a length-dependent interaction between the
two transitions~\cite{Gerard1994}. To characterize this interaction,
the additional fields in the uniform pinched section of radius $b_k$
are represented through its continuous cylindrical-mode spectrum.
Since the radius change is independent of the azimuthal angle $\phi$,
the $m=1$ dependence of the incident HE$_{11}$ field is preserved, so
only the corresponding $m=1$ cylindrical modes are required. Thus,
denoting by $q\geq0$ the transverse wavenumber in the surrounding air,
the relation $q^2+\beta_q^2=k_0^2$ gives
\begin{equation}
\small
\beta_q
=
\begin{cases}
\sqrt{k_0^2-q^2},
&
0\leq q<k_0,
\\[1mm]
-j\sqrt{q^2-k_0^2},
&
q>k_0,
\end{cases}
\label{eq:beta-q}
\end{equation}
where $q<k_0$ and $q>k_0$ correspond to longitudinally propagating and
decaying components, respectively, while the corresponding
dimensionless transverse wavenumber inside the pinched section is
$u_q=b_k\sqrt{k_0^2\epsilon_r-\beta_q^2}$. Following the
cylindrical radiation-mode representation in~\cite{LiBrown2006}, for
each $q$ the $m=1$ field is described by two independent cylindrical
solutions, indexed by $s\in\{1,2\}$, whose longitudinal electric and
magnetic fields within the uniform pinched section are{\small
\begin{align}
E_{z,q,s}(\rho,\phi,z;b_k)
&=
\Psi_{E,q,s}(\rho;b_k)
\cos\phi\,e^{-j\beta_qz},
\label{eq:Ez-q}
\\
H_{z,q,s}(\rho,\phi,z;b_k)
&=
\Psi_{H,q,s}(\rho;b_k)
\sin\phi\,e^{-j\beta_qz},
\label{eq:Hz-q}
\end{align}
}where
\begin{equation}
\small
\Psi_{E,q,s}(\rho;b_k)
=
\begin{cases}
C_{q,s}J_1\!\left(u_q\rho/b_k\right),
&
0\leq\rho<b_k,
\\[1mm]
A_{q,s}J_1(q\rho)
+
B_{q,s}Y_1(q\rho),
&
\rho\geq b_k,
\end{cases}
\label{eq:Psi-E-q}
\end{equation}
and
\begin{equation}
\small
\Psi_{H,q,s}(\rho;b_k)
=
\begin{cases}
D_{q,s}J_1\!\left(u_q\rho/b_k\right),
&
0\leq\rho<b_k,
\\[1mm]
F_{q,s}J_1(q\rho)
+
G_{q,s}Y_1(q\rho),
&
\rho\geq b_k,
\end{cases}
\label{eq:Psi-H-q}
\end{equation}
where $Y_1(\cdot)$ denotes the Bessel function of the second kind,
while the coefficients $C_{q,s}$, $D_{q,s}$, $A_{q,s}$, $B_{q,s}$,
$F_{q,s}$, and $G_{q,s}$ satisfy continuity of $E_z$, $H_z$,
$E_\phi$, and $H_\phi$ at $\rho=b_k$~\cite{LiBrown2006}. These
longitudinal fields then determine the transverse fields directly from
Maxwell's equations, consistently with \eqref{eq:Et-he11} and
\eqref{eq:Ht-he11}, as{\small
\begin{align}
\mathbf E_{t,q,s}(\rho,\phi,z;b_k)
&=
-\frac{j\left[
\beta_q\nabla_tE_{z,q,s}
-
\omega\mu_0
\widehat{\mathbf z}\times
\nabla_tH_{z,q,s}
\right]}{k_{t,q}^{2}(\rho)},
\label{eq:Et-q}
\\
\mathbf H_{t,q,s}(\rho,\phi,z;b_k)
&=
-\frac{j\left[
\beta_q\nabla_tH_{z,q,s}
+
\omega\varepsilon(\rho;b_k)
\widehat{\mathbf z}\times
\nabla_tE_{z,q,s}
\right]}{k_{t,q}^{2}(\rho)},
\label{eq:Ht-q}
\end{align}
}where
$\nabla_t=
\widehat{\boldsymbol\rho}\frac{\partial}{\partial\rho}
+
\widehat{\boldsymbol\phi}\frac{1}{\rho}
\frac{\partial}{\partial\phi}$,
and
\begin{equation}
\small
k_{t,q}^{2}(\rho)
=
\begin{cases}
u_q^2/b_k^2,
&
0\leq\rho<b_k,
\\[1mm]
q^2,
&
\rho\geq b_k.
\end{cases}
\label{eq:kt-q}
\end{equation}
To use these fields to determine how strongly each cylindrical
component is excited at a radius transition, their arbitrary modal scaling must first be fixed. Accordingly, the cylindrical solutions are scaled using the radiation-mode normalization in \cite{LiBrown2006}, while both the propagating and decaying components of the continuous spectrum are retained consistently with the radiation-spectrum treatment in \cite{Khalil2003}, and the resulting transverse fields are denoted by $\overline{\mathbf E}_{t,q,s}$ and $\overline{\mathbf H}_{t,q,s}$. Following the reciprocity relations for waveguide modes in \cite{Svendsen2013} and the radiation-spectrum treatment of dielectric-waveguide discontinuities in \cite{Khalil2003}, their coupling to the guided HE$_{11}$ field is obtained through the corresponding modal projection. Thus, the corresponding cross-overlaps are defined as{\small
\begin{align}
W_{k,s}(q)
&=\!
\!\frac{1}{2}\!
\int_{0}^{2\pi}\!\!\!\!
\int_{0}^{\infty}\!\!
\Big[
\overline{\mathbf E}_{t}^{+}(\rho,\phi,0;a) \!\times\!
\overline{\mathbf H}_{t,q,s}(\rho,\phi,0;b_k)
\!\Big]\!\!
\cdot\!\widehat{\mathbf z}
\rho{\rm d}\rho{\rm d}\phi ,
\label{eq:W-continuous}
\\
Q_{k,s}(q)
&=\!
\frac{1}{2}\!
\int_{0}^{2\pi}\!\!\!\!
\int_{0}^{\infty}\!\!
\Big[
\overline{\mathbf E}_{t,q,s}(\rho,\phi,0;b_k) \! \times \!
\overline{\mathbf H}_{t}^{+}(\rho,\phi,0;a)
\!\Big]\!\!
\cdot\!\widehat{\mathbf z}
\rho{\rm d}\rho{\rm d}\phi .
\label{eq:Q-continuous}
\end{align}
}The common factor $-j$ in (25) and (26) is absorbed into the
modal scaling used in (28), (29), and the radiation-mode
normalization of [27], making the transverse modal profiles
real at $z=0$ for propagating components. This phase convention,
retained throughout (30)--(34), ensures consistency between
the non-conjugated reciprocity-based projections in (28) and
(29) and the conjugated HE$_{11}$ projections in (16)
and (17). Moreover, since the HE$_{11}$ mode and the cylindrical
modes of the same uniform radius-$b_k$ section are mutually
orthogonal under the adopted modal projection \cite{Svendsen2013},
their corresponding cross-projections vanish. Consequently,
projecting the electric- and magnetic-field continuity relations
onto the cylindrical modes under the same forward--backward
convention yields the corresponding sum and difference of
spectral amplitudes~\cite{Khalil2003,Svendsen2013}. Hence,
for unit HE$_{11}$ incidence from the unperturbed section,
the cylindrical component excited at the first transition is
\begin{equation}
\small
c_{k,s}^{(1)}(q)
=
\frac{1}{2}
\left[
\left(1+a_{r,k}\right)W_{k,s}(q)
+
\left(1-a_{r,k}\right)Q_{k,s}(q)
\right].
\label{eq:c-first}
\end{equation}
Similarly, for unit HE$_{11}$ incidence from within the pinched
section, reversing the radius transition and propagation direction gives a transmitted HE$_{11}$ field in the unperturbed section with amplitude $a_{t,k}$ and a reflected HE$_{11}$ field in the pinched section with amplitude $-a_{r,k}$. Since the incident and reflected HE$_{11}$ fields supported by the $b_k$-radius section have vanishing cross-projections with the cylindrical modes, the corresponding
projection retains the transmitted contribution. Thus, using the same forward--backward convention gives
\begin{equation}
c_{k,s}^{(2)}(q)
=
\frac{a_{t,k}}{2}
\left[
W_{k,s}(q)
-
Q_{k,s}(q)
\right].
\label{eq:c-second}
\end{equation}
Thus, $c_{k,s}^{(1)}(q)$ and $c_{k,s}^{(2)}(q)$ quantify the coupling of the $(q,s)$-th cylindrical component associated with HE$_{11}$ incidence from the unperturbed and pinched sections, respectively, while reciprocity gives the same coefficients for the corresponding conversion back into the guided fields~\cite{Svendsen2013}. Once
excited, each cylindrical component propagates or decays between the
two radius transitions according to
$e^{-j\beta_qL_{{\rm P},k}}$ and can therefore couple back into the
HE$_{11}$ field at the opposite transition, consistently with the
radiation-spectrum treatment of successive dielectric-waveguide
discontinuities in~\cite{Gerard1994}. Consequently, combining the coupling at one transition, propagation through the finite pinched section, and coupling at the other transition over the cylindrical spectrum gives
\begin{align}
\small
\mu_k
&=
\sum_{s=1}^{2}
\int_{0}^{\infty}
\left[c_{k,s}^{(1)}(q)\right]^2
e^{-j\beta_qL_{{\rm P},k}}
\,{\rm d}q,
\label{eq:mu-k}
% \end{equation}
\\
% \begin{equation}
\small
\nu_k
&=
\sum_{s=1}^{2}
\int_{0}^{\infty}
c_{k,s}^{(1)}(q)
c_{k,s}^{(2)}(q)
e^{-j\beta_qL_{{\rm P},k}}
\,{\rm d}q,
\label{eq:nu-k}
% \end{equation}
% and
\\
% \begin{equation}
\small
\gamma_k
&=
\sum_{s=1}^{2}
\int_{0}^{\infty}
\left[c_{k,s}^{(2)}(q)\right]^2
e^{-j\beta_qL_{{\rm P},k}}
\,{\rm d}q,
\label{eq:gamma-k}
\end{align}
where $\mu_k$ accounts for the coupling between the HE$_{11}$ states
supported by the two unperturbed $a$-radius sections, $\nu_k$ accounts
for the coupling between an HE$_{11}$ state supported by an
unperturbed $a$-radius section and an HE$_{11}$ state supported by the
pinched $b_k$-radius section, while $\gamma_k$ accounts for the
coupling between the HE$_{11}$ states supported by the pinched
$b_k$-radius section at the two radius transitions. Hence, the local coefficients $a_{t,k}$ and $a_{r,k}$ characterize each individual radius transition, whereas $\mu_k$, $\nu_k$, and $\gamma_k$ characterize the coupling between the HE$_{11}$ states at the two transitions through the cylindrical spectrum and its dependence on $L_{{\rm P},k}$. The resulting guided-state interaction is then accounted for self-consistently through the internal HE$_{11}$ amplitudes, thereby determining the complex signal reflected toward the source and transmitted through the PA toward the downstream waveguide section, as stated in the following proposition.
\begin{proposition}
For the $k$-th PA characterized by the pinching ratio
$\kappa_k=b_k/a$ and the pinched-section length $L_{{\rm P},k}$,
the complex reflection and transmission coefficients are given by

{\small
\begin{align}
S_{11,k}^{\rm PA}
&=
a_{r,k}
+
\frac{
p_k
\left[
-a_{r,k}p_k
\left(
a_{t,k}^2+\nu_k^2
\right)
+
2\left(1-p_k\gamma_k\right)a_{t,k}\nu_k
\right]
}{
\Delta_k
},
\label{eq:S11-PA}
\\
S_{21,k}^{\rm PA}
&=
\mu_k
+
\frac{
p_k
\left[
\left(1-p_k\gamma_k\right)
\left(
a_{t,k}^2+\nu_k^2
\right)
-
2a_{r,k}p_ka_{t,k}\nu_k
\right]
}{
\Delta_k
},
\label{eq:S21-PA}
\end{align}
}respectively, where $\Delta_k
=
\left(1-p_k\gamma_k\right)^2
-
a_{r,k}^2p_k^2$, and
$p_k=e^{-j\beta(b_k)L_{{\rm P},k}}$
denotes the direct HE$_{11}$ propagation factor through the pinched
section. Consequently, the corresponding fractions of the incident
guided power reflected toward the source and transmitted toward the
downstream waveguide section are
$\left|S_{11,k}^{\rm PA}\right|^2$ and
$\left|S_{21,k}^{\rm PA}\right|^2$, respectively.
\label{prop:PA-Sparameters}
\end{proposition}

\begin{IEEEproof}
Consider a unit-amplitude HE$_{11}$ field incident on the PA and denote
by $x_{1,k}$ and $x_{2,k}$ the guided field amplitudes launched inside
the pinched section from the first and second transitions, respectively,
so that the corresponding fields arriving at the opposite transitions
are $p_kx_{1,k}$ and $p_kx_{2,k}$. Within the adopted modal representation, the cylindrical spectrum
provides the coupling between the guided HE$_{11}$ states associated
with the two radius transitions. In particular, propagation of the
spectrum excited by the incident HE$_{11}$ field and its coupling to
the internal guided state gives the contribution $\nu_k$, whereas the
corresponding coupling generated by an internal HE$_{11}$ field gives
$\gamma_k$. Accordingly, after the cylindrical components propagate or decay
between the two transitions and couple back into the guided field, the
internal guided-state interaction is described by $x_{1,k}$ and
$x_{2,k}$. Accounting for the spectrum-mediated coupling together with
the local transmission and internal reflection amplitudes $a_{t,k}$
and $-a_{r,k}$, these amplitudes satisfy the self-consistent relations $\left(1-p_k\gamma_k\right)x_{1,k}
+a_{r,k}p_kx_{2,k}=a_{t,k}$ and
$a_{r,k}p_kx_{1,k}
+\left(1-p_k\gamma_k\right)x_{2,k}=\nu_k$, which can be rewritten as{\small
\begin{align}
x_{1,k}
&=
\frac{
\left(1-p_k\gamma_k\right)a_{t,k}
-a_{r,k}p_k\nu_k
}{
\Delta_k
},
\label{eq:x1-k-solution}
\\
x_{2,k}
&=
\frac{
\left(1-p_k\gamma_k\right)\nu_k
-a_{r,k}p_ka_{t,k}
}{
\Delta_k
}.
\label{eq:x2-k-solution}
\end{align}
}In addition, the guided field transmitted through the second transition
consists of the coupling contribution $\mu_k$ between the two
unperturbed sections and the contributions associated with $x_{1,k}$
and $x_{2,k}$, yielding
\begin{equation}
\small
S_{21,k}^{\rm PA}
=
\mu_k
+
p_k
\left[
a_{t,k}x_{1,k}
+
\nu_kx_{2,k}
\right].
\label{eq:S21-PA-before}
\end{equation}
Therefore, substituting \eqref{eq:x1-k-solution} and
\eqref{eq:x2-k-solution} into \eqref{eq:S21-PA-before} gives
\eqref{eq:S21-PA}. Similarly, the reflected guided field consists of the
direct reflection $a_{r,k}$ and the contributions returning from the
pinched section through $x_{1,k}$ and $x_{2,k}$, yielding
\begin{equation}
\small
S_{11,k}^{\rm PA}
=
a_{r,k}
+
p_k
\left[
a_{t,k}x_{2,k}
+
\nu_kx_{1,k}
\right].
\label{eq:S11-PA-before}
\end{equation}
Thus, substituting \eqref{eq:x1-k-solution} and
\eqref{eq:x2-k-solution} into \eqref{eq:S11-PA-before} gives
\eqref{eq:S11-PA}, which concludes the proof.
\end{IEEEproof}
\begin{remark}
The pinched-section length $L_{{\rm P},k}$ determines both the
phase accumulated between the two radius transitions through
$p_k$ and the propagation or decay of the cylindrical components
through $\mu_k$, $\nu_k$, and $\gamma_k$. Consequently, varying
$L_{{\rm P},k}$ changes the coherent interaction between the fields
associated with the two transitions and, therefore, the amplitude
and phase of the transmitted guided signal. The resulting
representation concerns finite pinched sections with two
spatially separated radius transitions, rather than the
limit $L_{{\rm P},k}\to 0$.
\end{remark}
\begin{remark}
The adoption of the fundamental HE$_{11}$ mode as the guided
state follows from the considered single-mode operating regime
and does not exclude the radiative and evanescent fields excited
at the radius transitions. In particular, the propagation or
decay of these additional fields across the finite pinched
section and their subsequent recoupling into the guided
HE$_{11}$ state are incorporated through $\mu_k$, $\nu_k$,
and $\gamma_k$. Consequently, the scattering coefficients
characterize the complex reflected and transmitted guided
signals, while radiation constitutes an additional channel
through which the incident guided power is redistributed.
\end{remark}

\section{Far-Field Radiation and Pinching Antenna Gain}
\label{sec:pa-pattern}

As the guided signal reaches the \(k\)-th PA, the reduction of the
waveguide radius changes the propagation constant and the HE$_{11}$
field distribution over the finite pinched section, giving rise to
transmitted and reflected guided components together with radiation.
Although the scattering response characterizes the guided signal that
is reflected and transmitted by the PA, it does not determine the
angular and polarization characteristics of the radiated field, which
are required to characterize how the signals radiated by different PAs
combine in space. Since the pinched section itself constitutes the
radiating region, its far-field response must therefore be obtained
from the electric field established over the same finite geometry,
thereby preserving its dependence on the pinching ratio $\kappa_k$,
the pinched-section length $L_{{\rm P},k}$, and the corresponding
forward and backward HE$_{11}$ field components.

Relative to the unperturbed radius-$a$ reference waveguide, the reduced-radius section introduces a localized permittivity perturbation over the annular region $b_k\leq\rho\leq a$ and the finite length $L_{{\rm P},k}$. According to the volume-equivalence principle, this perturbation can be represented by an equivalent polarization-current distribution confined to this region, whose generated field outside the perturbation volume corresponds to the field radiated by the PA \cite{Balanis2012Advanced}. Accordingly, a local coordinate system is
centered at the $k$-th PA with its $z$-axis along the waveguide, while
a far-field observation point is written as
$\mathbf r=R\widehat{\mathbf u}(\theta,\phi)$, where $R$ denotes
the observation distance and $\widehat{\mathbf u}(\theta,\phi)$ the
observation direction. Using the electric field established within the
pinched section to determine the equivalent polarization current, the
corresponding radiation integral yields the complex vector far-field
response of the $k$-th PA, as stated in the following proposition.
\begin{proposition}
\label{prop:PA-far-field}
For the \(k\)-th PA characterized by the pinching ratio
\(\kappa_k=b_k/a\) and the pinched-section length \(L_{{\rm P},k}\),
the far-field electric field is given by
\begin{equation}
\small
\begin{aligned}
\mathbf E_{{\rm PA},k}^{\rm ff}(R,\theta,\phi)
={}&
\frac{e^{-jk_0R}}{R}
\Big[
F_{\theta,k}(\theta,\phi)\hat{\boldsymbol{\theta}}
+
F_{\phi,k}(\theta,\phi)\hat{\boldsymbol{\phi}}
\Big],
\end{aligned}
\label{eq:PA-far-field}
\end{equation}
where
\begingroup
\small
\begin{subequations}
\label{eq:PA-pattern-components}
\begin{flalign}
&F_{\theta,k}(\theta,\phi)
=
\frac{\omega^2\mu_0\epsilon_0(1-\epsilon_r)}{4}
\cos\phi
\Big[
A_k^{+}Z_k^{+}(\theta)R_{\theta,k}^{+}(\theta)
\notag\\[-1mm]
&\quad
-
A_k^{-}Z_k^{-}(\theta)R_{\theta,k}^{-}(\theta)
\Big],
&&
\label{eq:PA-pattern-theta}
\\[1mm]
&F_{\phi,k}(\theta,\phi)
=
\frac{\omega^2\mu_0\epsilon_0(1-\epsilon_r)}{4}
\sin\phi
\Big[
A_k^{+}Z_k^{+}(\theta)R_{\phi,k}^{+}(\theta)
\notag\\[-1mm]
&\quad
-
A_k^{-}Z_k^{-}(\theta)R_{\phi,k}^{-}(\theta)
\Big].
&&
\label{eq:PA-pattern-phi}
\end{flalign}
\end{subequations}
\endgroup
with $A_k^{+}=x_{1,k}$, and $A_k^{-}=p_kx_{2,k}$, \(Z_k^\pm(\theta)=
\bigl(e^{j[k_0\cos\theta\mp\beta(b_k)]L_{{\rm P},k}}-1\bigr)/
\bigl(j[k_0\cos\theta\mp\beta(b_k)]\bigr)\),
while the radial factors \(R_{\theta,k}^{\pm}(\theta)\) and
\(R_{\phi,k}^{\pm}(\theta)\) are equal to
\begingroup
\footnotesize
\begin{subequations}
\label{eq:PA-radial-factors}
\begin{flalign}
&R_{\theta,k}^{\pm}(\theta)
=
\frac{J_1\!\bigl(u(b_k)\bigr)}
{\sqrt{P_{\rm g}(b_k)}K_1\!\bigl(w(b_k)\bigr)}
\Bigg\{
\frac{\pm b_k^2}
{k_0w^2(b_k)(\epsilon_r-\cos^2\theta)}
\Bigg[
&&\notag\\[-1mm]
&\Big(
B_\theta\!(\theta)\omega\mu_0D(b_k)
-
A_\theta\!(\theta)\beta\!(b_k)C(b_k)\cos\theta
\Big)I_{1,k}
&&\notag\\[-1mm]
&\!\!+\!\Big(
B_\theta(\theta)\beta\!(b_k)C(b_k)
\!-\!
A_\theta\!(\theta)\omega\mu_0D(b_k)\cos\theta
\Big)I_{2,k}\!
\Bigg]\!+\!C(b_k)A_\theta\!(\theta)I_{0,k}\!
\Bigg\},
&&
\label{eq:PA-radial-theta}
\\[1mm]
&R_{\phi,k}^{\pm}(\theta)
=
\frac{J_1\!\bigl(u(b_k)\bigr)}
{\sqrt{P_{\rm g}(b_k)}K_1\!\bigl(w(b_k)\bigr)}
\Bigg\{
\frac{\mp b_k^2}
{k_0w^2(b_k)(\epsilon_r-\cos^2\theta)}
\Bigg[
&&\notag\\[-1mm]
&\Big(
B_\phi(\theta)\omega\mu_0D(b_k)
+
A_\phi(\theta)\beta(b_k)C(b_k)\cos\theta
\Big)I_{1,k}
&&\notag\\[-1mm]
&\!\!+\!\Big(
B_\phi\!(\theta)\beta\!(b_k)C\!(b_k)
\!+\!\!
A_\phi\!(\theta)\omega\mu_0D(b_k)\cos\theta
\Big)\!I_{2,k}\!
\Bigg]\!+\!C(b_k)A_\phi\!(\theta)I_{0,k}
\Bigg\}.
&&
\label{eq:PA-radial-phi}
\end{flalign}
\end{subequations}
\endgroup
Moreover, \(I_{0,k}\), \(I_{1,k}\), and
\(I_{2,k}\) are given by \eqref{eq:PA-I0}--\eqref{eq:PA-I2},
respectively, with
\(\left[f(\rho)\right]_{b_k}^{a}=f(a)-f(b_k)\), while
\(A_\nu(\theta)\) and \(B_\nu(\theta)\), for
\(\nu\in\{\theta,\phi\}\), denote the amplitudes of the reciprocal
field inside the waveguide and are given by
\eqref{eq:PA-reciprocal-A} and \eqref{eq:PA-reciprocal-B},
respectively. Finally, $H_1^{(2)}(\cdot)$ denotes the first-order
Hankel function of the second kind.
\end{proposition}
\begin{figure*}[!t]
\small
\begin{subequations}
\small
\label{eq:PA-radial-overlaps}
\begin{align}
&I_{0,k}=\Biggl[
\frac{\rho}
{k_0^2(\epsilon_r-\cos^2\theta)
+w^2(b_k)/b_k^2}
\Bigg\{
k_0\sqrt{\epsilon_r-\cos^2\theta}
J_2\!\left(
k_0\rho\sqrt{\epsilon_r-\cos^2\theta}
\right)
K_1\!\left(\frac{w(b_k)}{b_k}\rho\right) \notag
\\&-
\frac{w(b_k)J_1\!\left(
k_0\rho\sqrt{\epsilon_r-\cos^2\theta}
\right)
K_2\!\left(\frac{w(b_k)}{b_k}\rho\right)}{b_k}
\Bigg\}
\Biggr]_{b_k}^{a},
\label{eq:PA-I0}
\\
&I_{1,k}=
\left[
\rho
K_1\!\left(\frac{w(b_k)}{b_k}\rho\right)
\frac{\partial}{\partial\rho}
J_1\!\left(
k_0\rho\sqrt{\epsilon_r-\cos^2\theta}
\right)
\right]_{b_k}^{a}+
k_0^2(\epsilon_r-\cos^2\theta)I_{0,k},
\label{eq:PA-I1}
\\
&I_{2,k}=
\left[
K_1\!\left(\frac{w(b_k)}{b_k}\rho\right)
J_1\!\left(
k_0\rho\sqrt{\epsilon_r-\cos^2\theta}
\right)
\right]_{b_k}^{a}.
\label{eq:PA-I2}
\end{align}
\end{subequations}
\hrulefill
\begin{equation}
\small
\begin{split}
A_\nu(\theta)
={}&
-\frac{2\Bigg[
J_1'\!\left(k_0a\sin\theta\right)
-
\frac{
J_1\!\left(k_0a\sin\theta\right)
H_1^{(2)\prime}\!\left(k_0a\sin\theta\right)
}{
H_1^{(2)}\!\left(k_0a\sin\theta\right)
}
\Bigg]}{j\Delta(\theta)}
\times
\begin{cases}
\frac{
J_1'\!\left(
k_0a\sqrt{\epsilon_r-\cos^2\theta}
\right)
}{
\sqrt{\epsilon_r-\cos^2\theta}
}
-
\frac{
J_1\!\left(
k_0a\sqrt{\epsilon_r-\cos^2\theta}
\right)
H_1^{(2)\prime}\!\left(k_0a\sin\theta\right)
}{
\sin\theta\,
H_1^{(2)}\!\left(k_0a\sin\theta\right)
},
& \nu=\theta,
\\
\frac{
\cos\theta\,
J_1\!\left(
k_0a\sqrt{\epsilon_r-\cos^2\theta}
\right)
}{
ak_0
}
\left(
\frac{1}{\epsilon_r-\cos^2\theta}
-
\frac{1}{\sin^2\theta}
\right),
& \nu=\phi.
\end{cases}
\end{split}
\label{eq:PA-reciprocal-A}
\end{equation}
\begin{equation}
\small
\begin{split}
B_\nu(\theta)
={}&
-\frac{2\Bigg[
J_1'\!\left(k_0a\sin\theta\right)
-
\frac{
J_1\!\left(k_0a\sin\theta\right)
H_1^{(2)\prime}\!\left(k_0a\sin\theta\right)
}{
H_1^{(2)}\!\left(k_0a\sin\theta\right)
}
\Bigg]}{j\Delta(\theta)}
\times
\begin{cases}
\frac{
\cos\theta\,
J_1\!\left(
k_0a\sqrt{\epsilon_r-\cos^2\theta}
\right)
}{
ak_0
}
\left(
\frac{1}{\epsilon_r-\cos^2\theta}
-
\frac{1}{\sin^2\theta}
\right),
& \nu=\theta,
\\[5mm]
-\Bigg[
\frac{
\epsilon_r
J_1'\!\left(
k_0a\sqrt{\epsilon_r-\cos^2\theta}
\right)
}{
\sqrt{\epsilon_r-\cos^2\theta}
}
-
\frac{
J_1\!\left(
k_0a\sqrt{\epsilon_r-\cos^2\theta}
\right)
H_1^{(2)\prime}\!\left(k_0a\sin\theta\right)
}{
\sin\theta\,
H_1^{(2)}\!\left(k_0a\sin\theta\right)
}
\Bigg],
& \nu=\phi.
\end{cases}
\end{split}
\label{eq:PA-reciprocal-B}
\end{equation}
\begin{equation}
\small
\begin{split}
\Delta(\theta)
={}&
\Bigg\{
\frac{
\cos\theta\,
J_1\!\left(
k_0a\sqrt{\epsilon_r-\cos^2\theta}
\right)
}{
ak_0
}
\left(
\frac{1}{\epsilon_r-\cos^2\theta}
-
\frac{1}{\sin^2\theta}
\right)
\Bigg\}^{2}
-
\Bigg\{
\frac{
J_1'\!\left(
k_0a\sqrt{\epsilon_r-\cos^2\theta}
\right)
}{
\sqrt{\epsilon_r-\cos^2\theta}
}
\!\!-\!
\frac{
J_1\!\left(
k_0a\sqrt{\epsilon_r-\cos^2\theta}
\right)
H_1^{(2)\prime}\!\left(k_0a\sin\theta\right)
}{
\sin\theta\,
H_1^{(2)}\!\left(k_0a\sin\theta\right)
}
\Bigg\}
\\
&\times
\Bigg\{
\frac{
\epsilon_r
J_1'\!\left(
k_0a\sqrt{\epsilon_r-\cos^2\theta}
\right)
}{
\sqrt{\epsilon_r-\cos^2\theta}
}
\!-\!
\frac{
J_1\!\left(
k_0a\sqrt{\epsilon_r-\cos^2\theta}
\right)
H_1^{(2)\prime}\!\left(k_0a\sin\theta\right)
}{
\sin\theta\,
H_1^{(2)}\!\left(k_0a\sin\theta\right)
}
\Bigg\}.
\end{split}
\label{eq:PA-reciprocal-denominator}
\end{equation}
\hrulefill
\end{figure*}
\begin{IEEEproof}
 The proof can be found in Appendix A.
\end{IEEEproof}

Having characterized the complex far-field response of the $k$-th PA,
the corresponding radiation power-fraction and gain can be obtained directly
from its radiated-power distribution. In particular, using the standard
far-field power relation \cite{Balanis2016}, the radiation intensity of
the $k$-th PA in the direction $(\theta,\phi)$ is given by
$U_k(\theta,\phi)
=
\frac{1}{2\eta_0}
\left[
\left|F_{\theta,k}(\theta,\phi)\right|^2
+
\left|F_{\phi,k}(\theta,\phi)\right|^2
\right]$. Since the incident HE$_{11}$ mode is normalized to unit guided power,
integrating the radiation intensity over the complete solid angle
directly yields the fraction of the incident guided power radiated by
the $k$-th PA and, hence, its radiation power-fraction as
\begin{equation}
\small
\eta_{{\rm rad},k}
=
\int_{0}^{2\pi}
\int_{0}^{\pi}
U_k(\theta,\phi)
\sin\theta\,
{\rm d}\theta\,{\rm d}\phi.
\label{eq:PA-radiation-efficiency}
\end{equation}
Accordingly, the directivity of the $k$-th PA is obtained from the
angular distribution of the radiated power as $D_k(\theta,\phi)
=
\frac{4\pi U_k(\theta,\phi)}
{\eta_{{\rm rad},k}}$, while the corresponding gain of a PA is given by $G_k(\theta,\phi)
=
\eta_{{\rm rad},k}
D_k(\theta,\phi)$,
which can be rewritten as
\begin{equation}
\small
G_k(\theta,\phi)
=
\frac{2\pi}{\eta_0}
\left[
\left|F_{\theta,k}(\theta,\phi)\right|^2
+
\left|F_{\phi,k}(\theta,\phi)\right|^2
\right].
\label{eq:PA-gain-closed}
\end{equation}
Hence, the pinching ratio $\kappa_k$ and the
pinched-section length $L_{{\rm P},k}$ jointly determine not only the
fraction of guided power extracted and radiated by the PA through
$\eta_{{\rm rad},k}$, but also the angular redistribution of this power
through $D_k(\theta,\phi)$, with their combined effect ultimately
determining the PA gain $G_k(\theta,\phi)$, thus characterizing how an activated PA modifies the guided signal propagating within the waveguide and radiates the extracted power toward the wireless environment.

\section{Simulation Results}
In this section, the accuracy and validity of the proposed
electromagnetic model are evaluated through full-wave electromagnetic
simulations performed using MEEP \cite{Oskooi2010}.  Specifically, the considered lossless cylindrical waveguide is operated at $f_0=28$ GHz and is formed by a low-permittivity dielectric with $\epsilon_r=2.56$, consistent with experimentally reported values for flexible polymeric materials at millimeter-wave frequencies \cite{Hu2023}. Moreover, an unperturbed waveguide radius $a=2.75$ mm is assumed, whereas the pinching ratio is restricted to $0.75\leq\kappa_k\leq1$ to represent realistic and mechanically moderate perturbations consistent with the flexibility of the considered dielectric waveguide. Finally, the pinched-section length is varied as
$L_{{\rm P},k}/\lambda_0\in\{2,4,10\}$, allowing both relatively short and extended pinching regions to be represented, and the incident guided power is normalized to $P_t=1$ W.

\begin{figure*}[t]
    \centering

    % =========================================================
    % L_P,k = 2 lambda_0
    % =========================================================
    \begin{minipage}[t]{0.325\textwidth}
        \centering
        \begin{tikzpicture}
        \begin{axis}[
            width=0.98\linewidth,
            height=0.76\linewidth,
            title={$L_{{\rm P},k}=2\lambda_0$},
            xlabel={$\kappa$},
            ylabel={Power fraction},
            xmin=0.75,
            xmax=1,
            ymin=0,
            ymax=1,
            xtick={0.75,0.80,0.85,0.90,0.95,1.00},
            ytick={0,0.2,0.4,0.6,0.8,1.0},
            grid=major,
            grid style={dotted,gray!40},
            tick label style={font=\scriptsize},
            label style={font=\scriptsize},
            title style={font=\footnotesize},
            legend cell align={left},
            legend style={
                at={(0.97,0.55)},
                anchor=east,
                font=\scriptsize,
                draw=none
            }
        ]

        % |S21|^2
        \addplot[
            blue,
            line width=1.3pt,
            forget plot
        ]
        table[x=kappa,y=value]
        {MODEL_S21_power_L02lambda.dat};

        \addplot[
            blue,
            only marks,
            mark=*,
            mark size=1.5pt,
            forget plot
        ]
        table[x=kappa,y=value]
        {MEEP_S21_power_L02lambda.dat};

        % |S11|^2
        \addplot[
            red,
            line width=1.3pt,
            forget plot
        ]
        table[x=kappa,y=value]
        {MODEL_S11_power_L02lambda.dat};

        \addplot[
            red,
            only marks,
            mark=*,
            mark size=1.5pt,
            forget plot
        ]
        table[x=kappa,y=value]
        {MEEP_S11_power_L02lambda.dat};

        % Black legend
        \addlegendimage{black,line width=1.3pt}
        \addlegendentry{Analytical}

        \addlegendimage{black,only marks,mark=*,mark size=1.5pt}
        \addlegendentry{MEEP}

        \end{axis}
        \end{tikzpicture}
    \end{minipage}
    \hfill
    % =========================================================
    % L_P,k = 4 lambda_0
    % =========================================================
    \begin{minipage}[t]{0.325\textwidth}
        \centering
        \begin{tikzpicture}
        \begin{axis}[
            width=0.98\linewidth,
            height=0.76\linewidth,
            title={$L_{{\rm P},k}=4\lambda_0$},
            xlabel={$\kappa$},
            ylabel={Power fraction},
            xmin=0.75,
            xmax=1,
            ymin=0,
            ymax=1,
            xtick={0.75,0.80,0.85,0.90,0.95,1.00},
            ytick={0,0.2,0.4,0.6,0.8,1.0},
            grid=major,
            grid style={dotted,gray!40},
            tick label style={font=\scriptsize},
            label style={font=\scriptsize},
            title style={font=\footnotesize},
            legend cell align={left},
            legend style={
                at={(0.97,0.55)},
                anchor=east,
                font=\scriptsize,
                draw=none
            }
        ]

        % |S21|^2
        \addplot[
            blue,
            line width=1.3pt,
            forget plot
        ]
        table[x=kappa,y=value]
        {MODEL_S21_power_L04lambda.dat};

        \addplot[
            blue,
            only marks,
            mark=*,
            mark size=1.5pt,
            forget plot
        ]
        table[x=kappa,y=value]
        {MEEP_S21_power_L04lambda.dat};

        % |S11|^2
        \addplot[
            red,
            line width=1.3pt,
            forget plot
        ]
        table[x=kappa,y=value]
        {MODEL_S11_power_L04lambda.dat};

        \addplot[
            red,
            only marks,
            mark=*,
            mark size=1.5pt,
            forget plot
        ]
        table[x=kappa,y=value]
        {MEEP_S11_power_L04lambda.dat};

        % Black legend
        \addlegendimage{black,line width=1.3pt}
        \addlegendentry{Analytical}

        \addlegendimage{black,only marks,mark=*,mark size=1.5pt}
        \addlegendentry{MEEP}

        \end{axis}
        \end{tikzpicture}
    \end{minipage}
    \hfill
    % =========================================================
    % L_P,k = 10 lambda_0
    % =========================================================
    \begin{minipage}[t]{0.325\textwidth}
        \centering
        \begin{tikzpicture}
        \begin{axis}[
            width=0.98\linewidth,
            height=0.76\linewidth,
            title={$L_{{\rm P},k}=10\lambda_0$},
            xlabel={$\kappa$},
            ylabel={Power fraction},
            xmin=0.75,
            xmax=1,
            ymin=0,
            ymax=1,
            xtick={0.75,0.80,0.85,0.90,0.95,1.00},
            ytick={0,0.2,0.4,0.6,0.8,1.0},
            grid=major,
            grid style={dotted,gray!40},
            tick label style={font=\scriptsize},
            label style={font=\scriptsize},
            title style={font=\footnotesize},
            legend cell align={left},
            legend style={
                at={(0.97,0.55)},
                anchor=east,
                font=\scriptsize,
                draw=none
            }
        ]

        % |S21|^2
        \addplot[
            blue,
            line width=1.3pt,
            forget plot
        ]
        table[x=kappa,y=value]
        {MODEL_S21_power_L10lambda.dat};

        \addplot[
            blue,
            only marks,
            mark=*,
            mark size=1.5pt,
            forget plot
        ]
        table[x=kappa,y=value]
        {MEEP_S21_power_L10lambda.dat};

        % |S11|^2
        \addplot[
            red,
            line width=1.3pt,
            forget plot
        ]
        table[x=kappa,y=value]
        {MODEL_S11_power_L10lambda.dat};

        \addplot[
            red,
            only marks,
            mark=*,
            mark size=1.5pt,
            forget plot
        ]
        table[x=kappa,y=value]
        {MEEP_S11_power_L10lambda.dat};

        % Black legend
        \addlegendimage{black,line width=1.3pt}
        \addlegendentry{Analytical}

        \addlegendimage{black,only marks,mark=*,mark size=1.5pt}
        \addlegendentry{MEEP}

        \end{axis}
        \end{tikzpicture}
    \end{minipage}

    \caption{Scattering responses versus $\kappa$ for different values of $L_{{\rm P},k}$: (blue) $|S_{21}|^2$, (red) $|S_{11}|^2$.}
    \label{fig:Sparam_power_all_lengths}
\end{figure*}
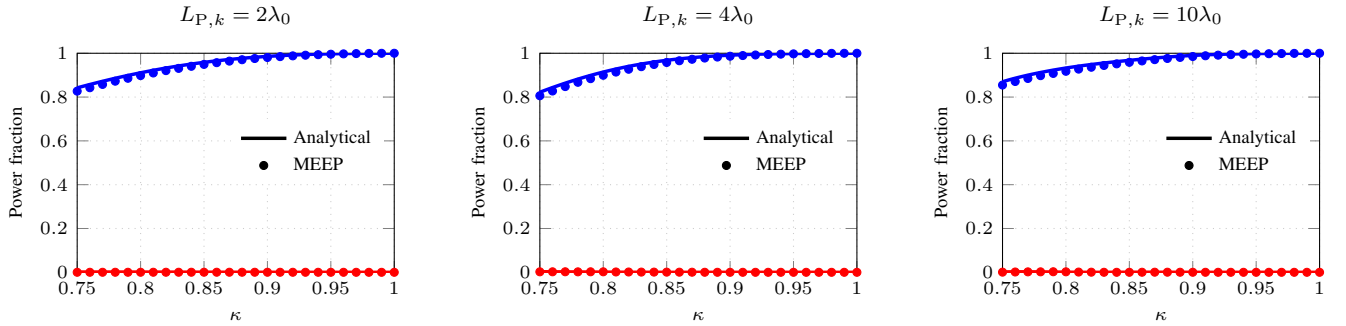

Fig.~\ref{fig:Sparam_power_all_lengths} illustrates the reflected and
transmitted power fractions as a function of the pinching ratio
$\kappa$ for different pinched-section lengths. In all examined cases,
the analytical results closely follow the MEEP results, validating the
accuracy of the proposed electromagnetic model. As $\kappa$ decreases,
the radius perturbation becomes stronger and $|S_{21}|^2$ decreases,
showing that a larger portion of the incident guided power is removed
from the propagating HE$_{11}$ mode, while $|S_{11}|^2$ remains below
$-25$ dB over the entire examined range, indicating that the reduction
of the transmitted power is associated with radiation. Importantly, the
transmitted power also exhibits a non-monotonic dependence on
$L_{{\rm P},k}$, as illustrated for $\kappa=0.75$, where increasing
$L_{{\rm P},k}$ from $2\lambda_0$ to $4\lambda_0$ reduces
$|S_{21}|^2$ from $0.842$ to $0.822$, while a further increase to
$10\lambda_0$ raises it to $0.870$, since the field propagating inside
the pinched section accumulates a length-dependent phase before being partially recoupled to the unperturbed waveguide at the second radius transition. Hence, the scattering response of a PA is characterized by its geometry, including both the pinching ratio and the PA length, and thus, an accurate characterization of the radiated power and the power remaining available along the waveguide requires the PA geometry.

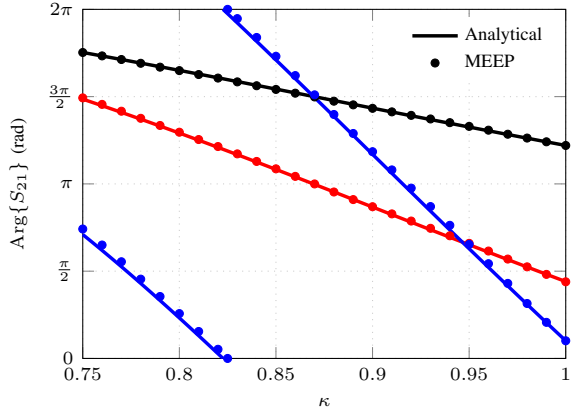
\begin{figure}[t]
    \centering
    \begin{tikzpicture}
    \begin{axis}[
        width=0.9\columnwidth,
        height=0.7\columnwidth,
        xlabel={$\kappa$},
        ylabel={$\mathrm{Arg}\{S_{21}\}$ (rad)},
        xmin=0.75,
        xmax=1,
        ymin=0,
        ymax=6.28318530718,
        xtick={0.75,0.80,0.85,0.90,0.95,1.00},
        ytick={
            0,
            1.57079632679,
            3.14159265359,
            4.71238898038,
            6.28318530718
        },
        yticklabels={
            $0$,
            $\frac{\pi}{2}$,
            $\pi$,
            $\frac{3\pi}{2}$,
            $2\pi$
        },
        grid=major,
        grid style={dotted,gray!40},
        tick label style={font=\scriptsize},
        label style={font=\scriptsize},
        unbounded coords=jump,
        legend cell align={left},
        legend style={
            at={(0.98,0.98)},
            anchor=north east,
            font=\scriptsize,
            draw=none
        }
    ]

    % L_P,k = 2 lambda_0
    \addplot[
        black,
        line width=1.3pt,
        forget plot
    ]
    table[x=kappa,y=value]
    {MODEL_S21_phase_rad_L02lambda.dat};

    \addplot[
        black,
        only marks,
        mark=*,
        mark size=1.5pt,
        forget plot
    ]
    table[x=kappa,y=value]
    {MEEP_S21_phase_rad_L02lambda.dat};

    % L_P,k = 4 lambda_0
    \addplot[
        red,
        line width=1.3pt,
        forget plot
    ]
    table[x=kappa,y=value]
    {MODEL_S21_phase_rad_L04lambda.dat};

    \addplot[
        red,
        only marks,
        mark=*,
        mark size=1.5pt,
        forget plot
    ]
    table[x=kappa,y=value]
    {MEEP_S21_phase_rad_L04lambda.dat};

    % L_P,k = 10 lambda_0
    \addplot[
        blue,
        line width=1.3pt,
        forget plot
    ]
    table[x=kappa,y=value]
    {MODEL_S21_phase_rad_L10lambda.dat};

    \addplot[
        blue,
        only marks,
        mark=*,
        mark size=1.5pt,
        forget plot
    ]
    table[x=kappa,y=value]
    {MEEP_S21_phase_rad_L10lambda.dat};

    % Legend only for line style
    \addlegendimage{black,line width=1.3pt}
    \addlegendentry{Analytical}

    \addlegendimage{black,only marks,mark=*,mark size=1.5pt}
    \addlegendentry{MEEP}

    \end{axis}
    \end{tikzpicture}

    \caption{$\mathrm{Arg}\{S_{21}\}$ versus $\kappa$ for different values of $L_{{\rm P},k}$: (black) $L_{{\rm P},k}=2\lambda_0$, (red) $L_{{\rm P},k}=4\lambda_0$, (blue) $L_{{\rm P},k}=10\lambda_0$.}
    \label{fig:S21_phase}
\end{figure}

Fig.~\ref{fig:S21_phase} illustrates the phase of $S_{21}$ as a function
of $\kappa$ for different pinched-section lengths. As it can be observed, the analytical results closely follow the MEEP results,
validating the phase characterization of the proposed electromagnetic
model. Additionally, it can be seen that as $\kappa$ varies, the propagation constant of the pinched section changes with respect to that of the unperturbed waveguide, so the guided signal accumulates a geometry-dependent phase while propagating through the PA, which is subsequently carried by the transmitted HE$_{11}$ mode along the downstream waveguide. This effect becomes more pronounced as $L_{{\rm P},k}$ increases, since the propagation-constant mismatch acts over a longer distance and therefore produces a larger phase excursion, with $L_{{\rm P},k}=10\lambda_0$ spanning the entire $[0,2\pi)$ range as $\kappa$ varies. Hence, the PA geometry does not only determine how much guided power remains after pinching, but also the phase state of the signal transmitted along the downstream waveguide, making the pinching ratio and pinched-section length relevant physical parameters for PAS transmission.

\begin{figure*}[t]
    \centering
    \scalebox{1}[0.75]{%
        \includegraphics[width=0.75\textwidth]{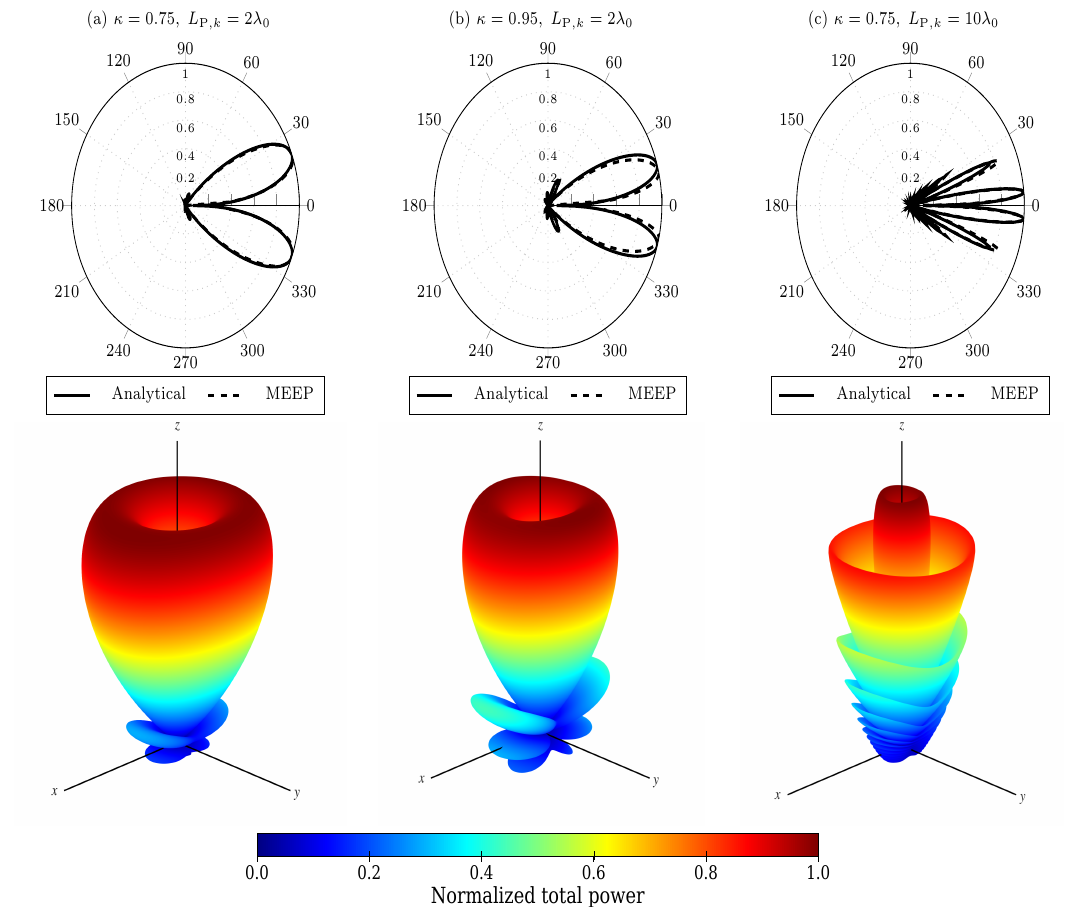}
    }
    \caption{Normalized far-field radiation patterns in the
$\phi=0^\circ$ principal plane, together with the corresponding 3D patterns, for (a) $\kappa=0.75$, $L_{{\rm P},k}=2\lambda_0$, (b) $\kappa=0.95$, $L_{{\rm P},k}=2\lambda_0$, and (c) $\kappa=0.75$, $L_{{\rm P},k}=10\lambda_0$}
    \label{fig:polar_patterns}
\end{figure*}

In Fig.~\ref{fig:polar_patterns}, the normalized far-field radiation
patterns are illustrated for representative PA geometries, where the
close agreement between the analytical and MEEP results validates the
ability of the proposed model to capture the angular redistribution of
the radiated power. In more detail, for the same pinched-section length, varying $\kappa$ modifies both the angular extent of the dominant radiation
region and the relative strength of the secondary lobes, whereas
increasing $L_{{\rm P},k}$ produces a considerably stronger reshaping
of the pattern. In particular, the longer pinched section introduces a
larger phase variation along the radiating region, causing its field
contributions to combine differently in the far field and resulting in
narrower angular features and a more pronounced sidelobe structure for
$L_{{\rm P},k}=10\lambda_0$. Consequently, the pinching ratio and pinched-section length determine not only the amount of radiated power but also its directional distribution, thereby acting as intrinsic physical parameters for shaping the PA radiation characteristics.

\begin{figure}
    \centering
    \begin{tikzpicture}
    \begin{axis}[
        width=0.9\columnwidth,
        height=0.7\columnwidth,
        xlabel={$\kappa$},
        ylabel={$\eta_{{\rm rad},k}$},
        xmin=0.75,
        xmax=1,
        ymin=0,
        ymax=0.30,
        xtick={0.75,0.80,0.85,0.90,0.95,1.00},
        ytick={0,0.05,0.10,0.15,0.20,0.25,0.30},
        yticklabel style={
            /pgf/number format/fixed,
            /pgf/number format/precision=2
        },
        grid=major,
        grid style={dotted,gray!40},
        tick label style={font=\scriptsize},
        label style={font=\scriptsize},
        unbounded coords=jump,
        legend cell align={left},
        legend style={
            at={(0.98,0.98)},
            anchor=north east,
            font=\scriptsize,
            draw=none
        }
    ]

    % =========================================================
    % L_P,k = 2 lambda_0
    % =========================================================

    \addplot[
        black,
        line width=1.3pt,
        forget plot
    ]
    table[
        x=kappa,
        y=value
    ]
    {MODEL_efficiency_farfield_L02lambda.dat};

    \addplot[
        black,
        only marks,
        mark=*,
        mark size=1.5pt,
        forget plot
    ]
    table[
        x=kappa,
        y=value
    ]
    {MEEP_efficiency_balance_L02lambda.dat};

    % =========================================================
    % L_P,k = 4 lambda_0
    % =========================================================

    \addplot[
        red,
        line width=1.3pt,
        forget plot
    ]
    table[
        x=kappa,
        y=value
    ]
    {MODEL_efficiency_farfield_L04lambda.dat};

    \addplot[
        red,
        only marks,
        mark=*,
        mark size=1.5pt,
        forget plot
    ]
    table[
        x=kappa,
        y=value
    ]
    {MEEP_efficiency_balance_L04lambda.dat};

    % =========================================================
    % L_P,k = 10 lambda_0
    % =========================================================

    \addplot[
        blue,
        line width=1.3pt,
        forget plot
    ]
    table[
        x=kappa,
        y=value
    ]
    {MODEL_efficiency_farfield_L10lambda.dat};

    \addplot[
        blue,
        only marks,
        mark=*,
        mark size=1.5pt,
        forget plot
    ]
    table[
        x=kappa,
        y=value
    ]
    {MEEP_efficiency_balance_L10lambda.dat};

    % Legend only for line style
    \addlegendimage{black,line width=1.3pt}
    \addlegendentry{Analytical}

    \addlegendimage{black,only marks,mark=*,mark size=1.5pt}
    \addlegendentry{MEEP}

    \end{axis}
    \end{tikzpicture}

    \caption{Radiation power-fraction versus $\kappa$ for different values of
    $L_{{\rm P},k}$: (black) $L_{{\rm P},k}=2\lambda_0$, (red)
    $L_{{\rm P},k}=4\lambda_0$, and (blue)
    $L_{{\rm P},k}=10\lambda_0$.}
    \label{fig:radiation_efficiency}
\end{figure}
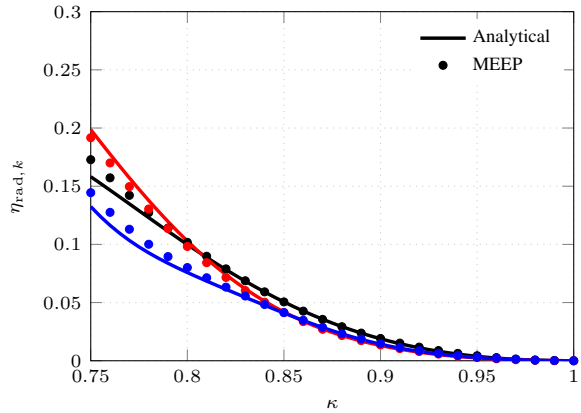

In Fig.~\ref{fig:radiation_efficiency}, the radiation power-fraction
$\eta_{{\rm rad},k}$ is illustrated as a function of $\kappa$ for
different pinched-section lengths. As $\kappa$ decreases, the resulting radius reduction enhances the coupling of the guided HE$_{11}$ field to radiation and therefore increases $\eta_{{\rm rad},k}$, whereas all curves converge to zero as $\kappa\rightarrow1$, consistently with the recovery of the unperturbed waveguide. At the same time, the radiation power-fraction
exhibits a non-monotonic dependence on $L_{{\rm P},k}$, since extending
the pinched section changes the relative phase accumulated by the
field contributions radiated along its finite length rather than
simply increasing the extracted power. Consequently, the pinched-section length yielding the largest radiation power-fraction among the examined values changes with the pinching ratio, further demonstrating that the radiation behavior of a PA is jointly characterized by the pinching ratio and pinched-section length.
\begin{figure}
    \centering
    \begin{tikzpicture}
    \begin{axis}[
        width=0.9\columnwidth,
        height=0.70\columnwidth,
        ylabel={$\gamma_r$ (dB)},
        xmin=0.4,
        xmax=3.6,
        ymin=0,
        ymax=14,
        xtick={1,2,3},
        xticklabels={,,},
        ytick={0,2,4,6,8,10,12,14},
        grid=major,
        grid style={dotted,gray!40},
        tick label style={font=\scriptsize},
        label style={font=\scriptsize},
        ybar,
        bar width=12pt,
        clip=false,
        nodes near coords={
            \pgfmathprintnumber[
                fixed,
                precision=2
            ]{\pgfplotspointmeta}
        },
        every node near coord/.append style={
            font=\scriptsize,
            anchor=south,
            yshift=1pt
        },
        legend cell align={left},
        legend columns=3,
        legend style={
            at={(0.5,1.03)},
            anchor=south,
            font=\scriptsize,
            draw=none,
            column sep=5pt
        }
    ]

    % Case 1
    \addplot[
        fill=orange!80,
        draw=black
    ]
    coordinates {(1,9.4519)};
    \addlegendentry{Case 1}

    % Case 2
    \addplot[
        fill=green!60!black!20,
        draw=black
    ]
    coordinates {(2,0.5091)};
    \addlegendentry{Case 2}

    % Case 3
    \addplot[
        fill=cyan!45,
        draw=black
    ]
    coordinates {(3,11.8274)};
    \addlegendentry{Case 3}

    \end{axis}
    \end{tikzpicture}

    \caption{Received SNR for the considered transmission cases.}
    \label{fig:system_snr_comparison}
\end{figure}
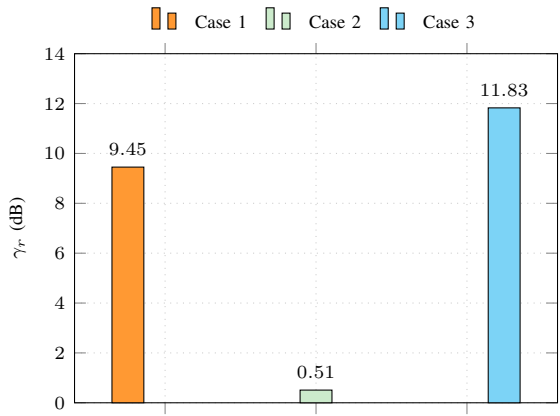

Finally, Fig.~\ref{fig:system_snr_comparison} shows the impact of the PA radiation characteristics on the received SNR, where, following the deployment considered in \cite{TyrovolasPerformancePAS}, a $30\times10$~m$^2$ service area is considered, with the dielectric waveguide extending along the center of the ceiling from $(0,5,3)$~m to $(30,5,3)$~m and $10^4$ users uniformly distributed over the floor. The local PA $z$-axis is aligned with the waveguide direction, corresponding to the global $x$-axis of the service area, while the local $x$-axis, which defines the polarization of the incident HE$_{11}$ mode, is aligned with the vertical global $z$-axis. The received SNR is evaluated as $\gamma_r=\gamma_t G(\theta,\phi)\left(\frac{\lambda_0}{4\pi R}\right)^2$, with $\gamma_t=10^9$ describing the transmit SNR, $R$ denoting the transmitter--user distance, $G(\theta,\phi)$ representing the gain referred to the incident guided power obtained from the proposed PA model, while the values reported in  Fig.~\ref{fig:system_snr_comparison} correspond to the received SNR averaged over the $10^4$ user locations. In particular, \textit{Case~1} considers an isotropic antenna located at the access point $(0,5,3)$~m, while \textit{Case~2} follows the distance-minimizing PA placement considered in~\cite{TyrovolasPerformancePAS}, such that the PA is positioned at the waveguide point closest to the user and $\kappa$ and $L_{{\rm P},k}$ are selected to maximize the radiated-power fraction $\eta_{{\rm rad},k}$ independently of the user direction. In contrast, \textit{Case~3} jointly optimizes the PA position, $\kappa$, and $L_{{\rm P},k}$ according to the received SNR through a dense numerical search over the admissible PA positions and all considered combinations of $\kappa\in[0.75,0.995]$ and $L_{{\rm P},k}/\lambda_0\in[2,10]$, with respective steps of $0.005$ and $0.05$. As it can be observed, even if \textit{Case~2} minimizes the PA--user distance and maximizes the total radiated-power fraction, it remains $8.94$~dB below \textit{Case~1}, since the resulting PA radiates only a fraction of the incident guided power and neither criterion accounts for how this power is angularly distributed toward the intended user. However, by jointly adapting the PA position and geometry, \textit{Case~3} additionally accounts for the corresponding directional gain, yielding improvements of $2.38$~dB and $11.32$~dB over \textit{Cases~1} and~\textit{2}, respectively. Therefore, the results illustrate how the geometry-dependent electromagnetic response provided by the proposed model can be incorporated into system-level PAS design, since the pinching ratio and pinched-section length jointly determine the radiated-power fraction, its directional distribution, and, consequently, the effectiveness of a given PA position.

\section{Conclusion}

In this paper, we developed an electromagnetic model that directly relates the finite perturbation applied to a dielectric waveguide to the guided-wave and radiation responses of the resulting PA. The proposed formulation jointly characterizes the complex reflection and transmission coefficients and the three-dimensional radiation characteristics, and its accuracy was validated through full-wave MEEP simulations. The results showed that the pinching ratio and pinched-section length jointly control the transmitted power and phase, radiation power-fraction, and radiation pattern, with their interaction producing non-monotonic responses and establishing the PA geometry as a physical design degree of freedom for shaping the radiation characteristics. Moreover, the considered single-PA received-SNR evaluation demonstrated that minimizing the PA--user distance or maximizing the total radiation power-fraction does not necessarily provide the highest received SNR, since the applied perturbation also determines the directional gain and, consequently, the effectiveness of a given PA position. Therefore, the proposed framework establishes a direct connection between the finite deformation-based PA geometry and PAS design, and can provide a basis for extending physics-compliant PAS modeling to multi-PA deployments, multi-mode operation, and other physically realizable waveguide perturbations.
\appendices

\section{Proof of Proposition 2}
The unperturbed radius-$a$ waveguide is adopted as the reference
structure, in which the incident HE${11}$ mode remains bound and does
not radiate. Hence, the radiation generated by the $k$-th PA originates
from the localized permittivity perturbation introduced by the radius
reduction. By representing this perturbation through the
volume-equivalent current $\mathbf J{{\rm eq},k}$, the additional
electric field produced by pinching at an observation point outside
the perturbed region can be expressed through the electric dyadic
Green function of the reference waveguide as~\cite{Balanis2012Advanced}
\begin{equation}
\small
\begin{split}
\mathbf E_{{\rm PA},k}(\mathbf r)
=
-j\omega\mu_0
\iiint_{\mathcal V_{{\rm P},k}}
\overline{\overline{\mathbf G}}_{\rm a}
(\mathbf r,\mathbf r')
\cdot
\mathbf J_{{\rm eq},k}(\mathbf r')
\,{\rm d}V',
\,\,
\mathbf r\notin\mathcal V_{{\rm P},k},
\end{split}
\label{eq:PA-equivalent-field}
\end{equation}
where $\mathbf r'$ denotes a source point,
$\overline{\overline{\mathbf G}}{\rm a}(\mathbf r,\mathbf r')$ is
the outgoing electric dyadic Green function of the reference waveguide,
and $\mathcal V_{{\rm P},k}$ denotes the annular perturbation volume
defined in the local cylindrical coordinates by $b_k\leq\rho\leq a$,
$0\leq\phi'<2\pi$, and $0\leq z\leq L_{{\rm P},k}$. For a far-field
observation point $\mathbf r=R\widehat{\mathbf u}(\theta,\phi)$,
with a fixed observation direction away from the waveguide axis,
the radiated field exhibits the outgoing dependence $e^{-jk_0R}/R$.
Therefore, by letting $\widehat{\boldsymbol{\nu}}\in
{\widehat{\boldsymbol{\theta}},
\widehat{\boldsymbol{\phi}}}$, the corresponding far-field
component is defined as
\begin{equation}
\small
\begin{split}
F_{\nu,k}(\theta,\phi)
=
\lim_{R\rightarrow\infty}
\Big[
R e^{jk_0R}
\widehat{\boldsymbol{\nu}}\cdot
\mathbf E_{{\rm PA},k}
\big(
R\widehat{\mathbf u}(\theta,\phi)
\big)
\Big],
\,\,
\nu\in\{\theta,\phi\}.
\end{split}
\label{eq:PA-far-field-component}
\end{equation}
Defining the far-field Green tensor of the reference waveguide as
$\overline{\overline{\mathbf G}}_{\rm a}^{\infty}
(\theta,\phi;\mathbf r')
=
\lim_{R\rightarrow\infty}
R e^{jk_0R}
\overline{\overline{\mathbf G}}_{\rm a}
\left(
R\widehat{\mathbf u}(\theta,\phi),\mathbf r'
\right)$, then \eqref{eq:PA-far-field-component} can be rewritten as
\begin{equation}
\small
\begin{split}
F_{\nu,k}(\theta,\phi)
=
-j\omega\mu_0\!
\iiint_{\mathcal V_{{\rm P},k}}\!\!\!\!
\widehat{\boldsymbol{\nu}}\cdot
\overline{\overline{\mathbf G}}_{\rm a}^{\infty}
(\theta,\phi;\mathbf r')
\!\cdot\!
\mathbf J_{{\rm eq},k}(\mathbf r')
\,{\rm d}V'.
\end{split}
\label{eq:PA-far-field-green}
\end{equation}
Since reciprocity holds for the reference waveguide, the source and
observation roles in \eqref{eq:PA-far-field-green} can be interchanged.
Hence, according to the mixed-reciprocity relation, the far-field Green
tensor can be mapped to the local field established within the
reference waveguide under a reciprocal plane-wave illumination from
the observation direction as~\cite{Cakoni2018}
\begin{equation}
\small
\begin{split}
4\pi\,
\widehat{\boldsymbol{\nu}}\cdot
\overline{\overline{\mathbf G}}_{\rm a}^{\infty}
(\theta,\phi;\mathbf r')
\cdot
\mathbf J_{{\rm eq},k}(\mathbf r')
=
\mathbf J_{{\rm eq},k}(\mathbf r')
\cdot
\mathbf E_{{\rm a},\nu}^{\rm r}(\mathbf r'),
\end{split}
\label{eq:PA-mixed-reciprocity}
\end{equation}
where $\mathbf E_{{\rm a},\nu}^{\rm r}(\mathbf r')$ denotes the field
established in the radius-$a$ reference waveguide by a unit
$\nu$-polarized plane wave incident from the observation direction.
Within $\mathcal V_{{\rm P},k}$, the radius reduction changes the
permittivity from $\epsilon_r\epsilon_0$ to $\epsilon_0$,
and the volume-equivalence principle therefore gives
\cite{Balanis2012Advanced}
\begin{equation}
\small
\mathbf J_{{\rm eq},k}(\mathbf r')
=
j\omega\epsilon_0(1-\epsilon_r)
\mathbf E_{{\rm P},k}^{\rm tot}(\mathbf r'),
\qquad
\mathbf r'\in V_{{\rm P},k},
\label{eq:PA-equivalent-current}
\end{equation}
where $\mathbf E_{{\rm P},k}^{\rm tot}$ denotes the total electric field established over the perturbation region. Consequently, substituting
\eqref{eq:PA-equivalent-current} and
\eqref{eq:PA-mixed-reciprocity} into
\eqref{eq:PA-far-field-green} yields
\begin{equation}
\small
\begin{split}
F_{\nu,k}(\theta,\phi)
=\!
\frac{\omega^2\mu_0\epsilon_0(1-\epsilon_r)}{4\pi}
\!\!\!\iiint_{V_{{\rm P},k}} \!\!\!
\mathbf E_{{\rm P},k}^{\rm tot}(\mathbf r')\!
\cdot\!
\mathbf E^r_{a,\nu}(\mathbf r')
dV'.
\end{split}
\label{eq:PA-reciprocal-radiation}
\end{equation}
Thus, the far-field radiation pattern is reduced to the overlap between
the electric field established over the pinched section and the
reciprocal field established over the same volume by an excitation
incident from the observation direction.

Within the adopted single-guided-state representation, the electric field
over the finite radius-$b_k$ section is represented by the forward- and
backward-propagating HE$_{11}$ components supported by this uniform
section, whose amplitudes are jointly determined by the two radius
transitions and the propagation between them. In particular, the
cylindrical-mode spectrum excited at each discontinuity accounts for
the interaction between the two transitions through propagation or
decay across the finite pinched section and subsequent recoupling into
the guided HE$_{11}$ state, consistently with the modal treatment of
dielectric-waveguide discontinuities and successive discontinuities
\cite{Khalil2003,Gerard1994}. Accordingly, this interaction is incorporated
into the internal guided amplitudes $x_{1,k}$ and $x_{2,k}$ through
the coupling quantities $\nu_k$ and $\gamma_k$, such that the
interaction-corrected field used for the closed-form radiation
characterization is represented by the corresponding forward- and
backward-propagating HE$_{11}$ fields. Hence, using the normalized
HE$_{11}$ fields supported by the radius-$b_k$ section, the electric
field within the pinched section is expressed as
\begin{equation}
\small
\mathbf E_{{\rm P},k}(\rho,\phi',z)
\simeq
A_k^{+}
\overline{\mathbf E}_{b_k}^{+}(\rho,\phi',z)
+
A_k^{-}
\overline{\mathbf E}_{b_k}^{-}(\rho,\phi',z),
\label{eq:PA-pinched-section-field}
\end{equation}
where $A_k^{+}=x_{1,k}$ and $A_k^{-}=p_kx_{2,k}$. According to the adopted forward--backward waveguide-mode convention \cite{Svendsen2013}, for which
$\beta^{-}(b_k)=-\beta^{+}(b_k)$,
$C^{-}(b_k)=-C^{+}(b_k)$, and
$D^{-}(b_k)=D^{+}(b_k)$, the transverse electric-field profile is
preserved between the two propagation directions whereas the
longitudinal electric component changes sign, and therefore,
substituting \eqref{eq:psi-he11},
\eqref{eq:psi-he11-derivative},
\eqref{eq:Et-he11}, and
\eqref{eq:power-normalized-electric-field} into
\eqref{eq:PA-pinched-section-field} yields
\eqref{eq:PA-pinched-section-field-closed} for
$b_k\leq\rho\leq a$.

\begin{figure*}[!t]
\begin{equation}
\small
\begin{aligned}
&\mathbf E_{{\rm P},k}(\rho,\phi',z)
=
\frac{J_1(u(b_k))}
{\sqrt{P_{\rm g}(b_k)}K_1(w(b_k))}
\Bigg\{
\frac{jb_k^2}{w^2(b_k)}
\Big[
A_k^{+}e^{-j\beta(b_k)z}
+
A_k^{-}e^{j\beta(b_k)z}
\Big]
\Bigg[
\\[-1mm]
&\widehat{\boldsymbol\rho}\cos\phi'
\Bigg(
-\frac{\beta(b_k)C(b_k)w(b_k)}{b_k}
K_0\!\left(\frac{w(b_k)\rho}{b_k}\right)
+
\frac{\omega\mu_0D(b_k)-\beta(b_k)C(b_k)}{\rho}
K_1\!\left(\frac{w(b_k)\rho}{b_k}\right)
\Bigg)
\\[-1mm]
&+
\widehat{\boldsymbol\phi}\sin\phi'
\Bigg(
\frac{\omega\mu_0D(b_k)w(b_k)}{b_k}
K_0\!\left(\frac{w(b_k)\rho}{b_k}\right)
+
\frac{\omega\mu_0D(b_k)-\beta(b_k)C(b_k)}{\rho}
K_1\!\left(\frac{w(b_k)\rho}{b_k}\right)
\Bigg)
\Bigg]
\\[-1mm]
&+
\widehat{\mathbf z}\,
C(b_k)
K_1\!\left(\frac{w(b_k)\rho}{b_k}\right)
\cos\phi'
\Big[
A_k^{+}e^{-j\beta(b_k)z}
-
A_k^{-}e^{j\beta(b_k)z}
\Big]
\Bigg\}.
\end{aligned}
\label{eq:PA-pinched-section-field-closed}
\end{equation}
\vspace{-2mm}
\hrule
\vspace{-1mm}
\end{figure*}
The reciprocal field $\mathbf E_{{\rm a},\nu}^{\rm r}$ is established
when a unit $\nu$-polarized plane wave incident from the observation
direction impinges on the radius-$a$ reference waveguide. Following the
cylindrical-wave representation in~\cite{Tsalamengas1995}, this field
can be expanded in cylindrical harmonics, while the first-order
azimuthal dependence of $\mathbf E_{{\rm P},k}$ implies through
\eqref{eq:PA-reciprocal-radiation} that only the $m=\pm1$ harmonics
contribute to the overlap. Since the PA volume satisfies $\rho<a$, the
reciprocal field required in the radiation integral is its interior
field, whose transverse wavenumber is
$q_{\rm i}(\theta)=k_0\sqrt{\epsilon_r-\cos^2\theta}$, while
$q_{\rm o}(\theta)=k_0\sin\theta$ denotes the corresponding transverse
wavenumber in the surrounding air. Accordingly, the longitudinal
electric and magnetic components of the interior reciprocal field are
\begin{equation}
\small
E_{{\rm a},z,\nu}^{\rm r}
=
e^{jk_0z\cos\theta}
\begin{cases}
A_\theta(\theta)
J_1\!\left(q_{\rm i}(\theta)\rho\right)
\cos(\phi'-\phi),
&
\nu=\theta,
\\[2mm]
-A_\phi(\theta)
J_1\!\left(q_{\rm i}(\theta)\rho\right)
\sin(\phi'-\phi),
&
\nu=\phi,
\end{cases}
\label{eq:PA-reciprocal-longitudinal-E}
\end{equation}
and
\begin{equation}
\small
H_{{\rm a},z,\nu}^{\rm r}
=
\frac{e^{jk_0z\cos\theta}}{\eta_0}
\begin{cases}
B_\theta(\theta)
J_1\!\left(q_{\rm i}(\theta)\rho\right)
\sin(\phi'-\phi),
&
\nu=\theta,
\\[2mm]
-B_\phi(\theta)
J_1\!\left(q_{\rm i}(\theta)\rho\right)
\cos(\phi'-\phi),
&
\nu=\phi,
\end{cases}
\label{eq:PA-reciprocal-longitudinal-H}
\end{equation}
where $\eta_0$ denotes the free-space impedance. Once the longitudinal
components are specified, the corresponding transverse fields follow
directly from Maxwell's equations as{\small
\begin{align}
\mathbf E_{{\rm a},t,\nu}^{\rm r}
&=
\frac{j}{q_{\rm i}^{2}(\theta)}
\left[
k_0\cos\theta\,
\nabla_t E_{{\rm a},z,\nu}^{\rm r}
+
\omega\mu_0
\widehat{\mathbf z}\times
\nabla_t H_{{\rm a},z,\nu}^{\rm r}
\right],
\label{eq:PA-reciprocal-transverse-E}
\\
\mathbf H_{{\rm a},t,\nu}^{\rm r}
&=
\frac{j}{q_{\rm i}^{2}(\theta)}
\left[
k_0\cos\theta\,
\nabla_t H_{{\rm a},z,\nu}^{\rm r}
-
\omega\epsilon_r\epsilon_0
\widehat{\mathbf z}\times
\nabla_t E_{{\rm a},z,\nu}^{\rm r}
\right].
\label{eq:PA-reciprocal-transverse-H}
\end{align}
}Using
$\widehat{\mathbf z}\times\widehat{\boldsymbol\rho}
=\widehat{\boldsymbol\phi}$ together with
$\widehat{\mathbf z}\times\widehat{\boldsymbol\phi}
=-\widehat{\boldsymbol\rho}$, the transverse relations above give the
components required both for the boundary conditions and for the
radiation integral as
\begin{equation}
\small
\begin{aligned}
E_{{\rm a},\rho,\nu}^{\rm r}
={}&
\frac{j}{q_{\rm i}^{2}(\theta)}
\Bigg[
k_0\cos\theta
\frac{\partial E_{{\rm a},z,\nu}^{\rm r}}{\partial\rho}
-
\frac{\omega\mu_0}{\rho}
\frac{\partial H_{{\rm a},z,\nu}^{\rm r}}{\partial\phi'}
\Bigg],
\\
E_{{\rm a},\phi,\nu}^{\rm r}
={}&
\frac{j}{q_{\rm i}^{2}(\theta)}
\Bigg[
\frac{k_0\cos\theta}{\rho}
\frac{\partial E_{{\rm a},z,\nu}^{\rm r}}{\partial\phi'}
+
\omega\mu_0
\frac{\partial H_{{\rm a},z,\nu}^{\rm r}}{\partial\rho}
\Bigg],
\\
H_{{\rm a},\phi,\nu}^{\rm r}
={}&
\frac{j}{q_{\rm i}^{2}(\theta)}
\Bigg[
\frac{k_0\cos\theta}{\rho}
\frac{\partial H_{{\rm a},z,\nu}^{\rm r}}{\partial\phi'}
-
\omega\epsilon_r\epsilon_0
\frac{\partial E_{{\rm a},z,\nu}^{\rm r}}{\partial\rho}
\Bigg].
\end{aligned}
\label{eq:PA-reciprocal-components}
\end{equation}
The amplitudes $A_\nu(\theta)$ and $B_\nu(\theta)$ are then obtained
by matching the interior fields to the exterior incident and outgoing
cylindrical fields at $\rho=a$ and enforcing continuity of
$E_z$, $H_z$, $E_\phi$, and $H_\phi$~\cite{Tsalamengas1995}, which
gives \eqref{eq:PA-reciprocal-A} and
\eqref{eq:PA-reciprocal-B}. Consequently, substituting
\eqref{eq:PA-reciprocal-longitudinal-E} and
\eqref{eq:PA-reciprocal-longitudinal-H} into
\eqref{eq:PA-reciprocal-components}, while using the identity $\frac{\partial}{\partial\rho}
J_1\!\left(q_{\rm i}(\theta)\rho\right)
=
q_{\rm i}(\theta)
J_1'\!\left(q_{\rm i}(\theta)\rho\right)$ and $\omega\mu_0/\eta_0=k_0$, provides the radial and azimuthal components of the reciprocal electric field. Combining these components with the corresponding longitudinal component according to
\begin{equation}
\small
\mathbf E_{{\rm a},\nu}^{\rm r}
=
\widehat{\boldsymbol\rho}\,
E_{{\rm a},\rho,\nu}^{\rm r}
+
\widehat{\boldsymbol\phi}\,
E_{{\rm a},\phi,\nu}^{\rm r}
+
\widehat{\mathbf z}\,
E_{{\rm a},z,\nu}^{\rm r},
\qquad
\nu\in\{\theta,\phi\},
\label{eq:PA-reciprocal-total-E1}
\end{equation}
yields the explicit reciprocal electric field in
\eqref{eq:PA-reciprocal-total-E2}.

\begin{figure*}[!t]
\begin{equation}
\small
\mathbf E_{{\rm a},\nu}^{\rm r}
\!=\!
e^{jk_0z\cos\theta}
\begin{cases}
\begin{aligned}
&
\widehat{\mathbf z}\,
A_\theta(\theta)
J_1\!\left(q_{\rm i}(\theta)\rho\right)
\cos(\phi'-\phi)
\\
&+
\frac{
j\widehat{\boldsymbol\rho}\cos(\phi'-\phi)
}{
q_{\rm i}^{2}(\theta)
}
\Bigg[
k_0\cos\theta\,
A_\theta(\theta)q_{\rm i}(\theta)
J_1'\!\left(q_{\rm i}(\theta)\rho\right)
-
\frac{k_0B_\theta(\theta)}{\rho}
J_1\!\left(q_{\rm i}(\theta)\rho\right)
\Bigg]
\\
&\quad+
\frac{
j\widehat{\boldsymbol\phi}\sin(\phi'-\phi)
}{
q_{\rm i}^{2}(\theta)
}
\Bigg[
-\frac{
k_0\cos\theta A_\theta(\theta)
}{\rho}
J_1\!\left(q_{\rm i}(\theta)\rho\right)
+
k_0B_\theta(\theta)q_{\rm i}(\theta)
J_1'\!\left(q_{\rm i}(\theta)\rho\right)
\Bigg],
\end{aligned}
&
\nu=\theta,
\\[6mm]
\begin{aligned}
&
-\widehat{\mathbf z}\,
A_\phi(\theta)
J_1\!\left(q_{\rm i}(\theta)\rho\right)
\sin(\phi'-\phi)
\\
&-
\frac{
j\widehat{\boldsymbol\rho}\sin(\phi'-\phi)
}{
q_{\rm i}^{2}(\theta)
}
\Bigg[
k_0\cos\theta\,
A_\phi(\theta)q_{\rm i}(\theta)
J_1'\!\left(q_{\rm i}(\theta)\rho\right)
+
\frac{k_0B_\phi(\theta)}{\rho}
J_1\!\left(q_{\rm i}(\theta)\rho\right)
\Bigg]
\\
&\quad-
\frac{
j\widehat{\boldsymbol\phi}\cos(\phi'-\phi)
}{
q_{\rm i}^{2}(\theta)
}
\Bigg[
\frac{
k_0\cos\theta A_\phi(\theta)
}{\rho}
J_1\!\left(q_{\rm i}(\theta)\rho\right)
+
k_0B_\phi(\theta)q_{\rm i}(\theta)
J_1'\!\left(q_{\rm i}(\theta)\rho\right)
\Bigg].
\end{aligned}
&
\nu=\phi .
\end{cases}
\label{eq:PA-reciprocal-total-E2}
\end{equation}
\vspace{-2mm}
\hrule
\vspace{-1mm}
\end{figure*}

With $\mathbf E_{{\rm P},k}$ and
$\mathbf E_{{\rm a},\nu}^{\rm r}$ established over the same annular
region, their overlap in \eqref{eq:PA-reciprocal-radiation} can be
evaluated directly. Expressing the PA volume in cylindrical coordinates
gives
\begin{equation}
\small
\begin{split}
F_{\nu,k}(\theta,\phi)
={}&
\frac{\omega^2\mu_0\epsilon_0(1-\epsilon_r)}{4\pi}
\int_{0}^{L_{{\rm P},k}}
\Bigg\{
\int_{b_k}^{a}
\rho
\Bigg[
\int_{0}^{2\pi}
\\[-1mm]
&\mathbf E_{{\rm P},k}(\rho,\phi',z)
\cdot
\mathbf E_{{\rm a},\nu}^{\rm r}(\rho,\phi',z)
\,{\rm d}\phi'
\Bigg]
{\rm d}\rho
\Bigg\}
{\rm d}z .
\end{split}
\label{eq:PA-reciprocal-radiation-expanded}
\end{equation}

The integral in \eqref{eq:PA-reciprocal-radiation-expanded} is
evaluated separately for the two reciprocal polarizations
$\nu\in\{\theta,\phi\}$.
\begin{itemize}
\item[]
\underline{\(\nu=\theta\):}
For the $\theta$-polarized reciprocal field, substituting
\eqref{eq:PA-pinched-section-field-closed} and the
$\nu=\theta$ branch of
\eqref{eq:PA-reciprocal-total-E2} into
\eqref{eq:PA-reciprocal-radiation-expanded} yields
{\scriptsize
\allowdisplaybreaks[4]
\begin{align}
&F_{\theta,k}(\theta,\phi)
=
\frac{\omega^2\mu_0\epsilon_0(1-\epsilon_r)}{4\pi}
\frac{J_1(u(b_k))}
{\sqrt{P_{\rm g}(b_k)}K_1(w(b_k))}
\notag\\[-1mm]
&\times
\int_{0}^{L_{{\rm P},k}}
\int_{b_k}^{a}
\int_{0}^{2\pi}
\rho e^{jk_0z\cos\theta}
\Bigg\{
-\frac{b_k^2}
{w^2(b_k)q_{\rm i}^2(\theta)}
\notag\\[-1mm]
&\times
\Big[
A_k^{+}e^{-j\beta(b_k)z}
+
A_k^{-}e^{j\beta(b_k)z}
\Big]
\Bigg[
\cos\phi'\cos(\phi'-\phi)
\notag\\[-1mm]
&\times
\Bigg(
-\frac{\beta(b_k)C(b_k)w(b_k)}{b_k}
K_0\!\left(\frac{w(b_k)}{b_k}\rho\right)
\notag\\[-1mm]
&+
\frac{\omega\mu_0D(b_k)-\beta(b_k)C(b_k)}{\rho}
K_1\!\left(\frac{w(b_k)}{b_k}\rho\right)
\Bigg)
\notag\\[-1mm]
&\times
\Bigg(
k_0\cos\theta A_\theta(\theta)q_{\rm i}(\theta)
J_1'\!\left(q_{\rm i}(\theta)\rho\right)
-
\frac{k_0B_\theta(\theta)}{\rho}
J_1\!\left(q_{\rm i}(\theta)\rho\right)
\Bigg)
\notag\\[-1mm]
&+
\sin\phi'\sin(\phi'-\phi)
\Bigg(
\frac{\omega\mu_0D(b_k)w(b_k)}{b_k}
K_0\!\left(\frac{w(b_k)}{b_k}\rho\right)
\notag\\[-1mm]
&+
\frac{\omega\mu_0D(b_k)-\beta(b_k)C(b_k)}{\rho}
K_1\!\left(\frac{w(b_k)}{b_k}\rho\right)
\Bigg)
\notag\\[-1mm]
&\times
\Bigg(
-\frac{k_0\cos\theta A_\theta(\theta)}{\rho}
J_1\!\left(q_{\rm i}(\theta)\rho\right)
+
k_0B_\theta(\theta)q_{\rm i}(\theta)
J_1'\!\left(q_{\rm i}(\theta)\rho\right)
\Bigg)
\Bigg]
\notag\\[-1mm]
&+
C(b_k)A_\theta(\theta)
K_1\!\left(\frac{w(b_k)}{b_k}\rho\right)
J_1\!\left(q_{\rm i}(\theta)\rho\right)
\cos\phi'\cos(\phi'-\phi)
\notag\\[-1mm]
&\times
\Big[
A_k^{+}e^{-j\beta(b_k)z}
-
A_k^{-}e^{j\beta(b_k)z}
\Big]
\Bigg\}
\,{\rm d}\phi'\,{\rm d}\rho\,{\rm d}z .
\label{eq:PA-proof-theta-expanded}
\end{align}
}Since the dependence on $\phi'$ in
\eqref{eq:PA-proof-theta-expanded} is confined to
$\cos\phi'\cos(\phi'-\phi)$ and
$\sin\phi'\sin(\phi'-\phi)$, by utilizing
\cite[2.532.6]{Gradshteyn2014} and
\cite[2.532.4]{Gradshteyn2014},
\eqref{eq:PA-proof-theta-expanded} can be rewritten as
{\scriptsize
\allowdisplaybreaks[4]
\begin{align}
&F_{\theta,k}(\theta,\phi)
=
\frac{\omega^2\mu_0\epsilon_0(1-\epsilon_r)}{4}
\frac{J_1(u(b_k))\cos\phi}
{\sqrt{P_{\rm g}(b_k)}K_1(w(b_k))}
\notag\\
&\times
\int_{0}^{L_{{\rm P},k}}
e^{jk_0z\cos\theta}
\Bigg\{
-\frac{b_k^2k_0}
{w^2(b_k)q_{\rm i}^2(\theta)}
\Big[
A_k^{+}e^{-j\beta(b_k)z}
+
A_k^{-}e^{j\beta(b_k)z}
\Big]
\notag\\
&\times
\Bigg[
\frac{w(b_k)}{b_k}q_{\rm i}(\theta)
\Big[
B_\theta(\theta)\omega\mu_0D(b_k)
-
A_\theta(\theta)\beta(b_k)C(b_k)\cos\theta
\Big]
\notag\\
&\times
\int_{b_k}^{a}
\rho
K_0\!\left(\frac{w(b_k)}{b_k}\rho\right)
J_1'\!\left(q_{\rm i}(\theta)\rho\right)
\,{\rm d}\rho +
\frac{w(b_k)}{b_k}
\Big[
B_\theta(\theta)\beta(b_k)C(b_k)
\notag\\
&
-
A_\theta(\theta)\omega\mu_0D(b_k)\cos\theta
\Big]
\int_{b_k}^{a}
K_0\!\left(\frac{w(b_k)}{b_k}\rho\right)
J_1\!\left(q_{\rm i}(\theta)\rho\right)
\,{\rm d}\rho
\notag\\
&+
q_{\rm i}(\theta)
\Big[
\omega\mu_0D(b_k)-\beta(b_k)C(b_k)
\Big]
\Big[
A_\theta(\theta)\cos\theta+B_\theta(\theta)
\Big]
\notag\\
&\times
\int_{b_k}^{a}
K_1\!\left(\frac{w(b_k)}{b_k}\rho\right)
J_1'\!\left(q_{\rm i}(\theta)\rho\right)
\,{\rm d}\rho-
\Big[
\omega\mu_0D(b_k)-\beta(b_k)C(b_k)
\Big]
\notag\\
&\times
\Big[
A_\theta(\theta)\cos\theta+B_\theta(\theta)
\Big]
\int_{b_k}^{a}
\frac{1}{\rho}
K_1\!\left(\frac{w(b_k)}{b_k}\rho\right)
J_1\!\left(q_{\rm i}(\theta)\rho\right)
\,{\rm d}\rho
\Bigg]
\notag\\
&+
C(b_k)A_\theta(\theta)
\Big[
A_k^{+}e^{-j\beta(b_k)z}
-
A_k^{-}e^{j\beta(b_k)z}
\Big]
\notag\\
&\times
\int_{b_k}^{a}
\rho
K_1\!\left(\frac{w(b_k)}{b_k}\rho\right)
J_1\!\left(q_{\rm i}(\theta)\rho\right)
\,{\rm d}\rho
\Bigg\}
\,{\rm d}z .
\label{eq:PA-proof-theta-after-phi}
\end{align}
}Moreover, by using the relation
$K_\nu'(x)=-K_{\nu-1}(x)-\nu K_\nu(x)/x$,
given in \cite[8.486.12]{Gradshteyn2014}, with $\nu=1$, and
subsequently applying the chain rule to the resulting
Bessel-function derivative terms,
\eqref{eq:PA-proof-theta-after-phi} can be rewritten as
{\scriptsize
\allowdisplaybreaks[4]
\begin{align}
&F_{\theta,k}(\theta,\phi)
=
\frac{\omega^2\mu_0\epsilon_0(1-\epsilon_r)}{4}
\frac{J_1(u(b_k))\cos\phi}
{\sqrt{P_{\rm g}(b_k)}K_1(w(b_k))}
\notag\\[-1mm]
&\times
\int_{0}^{L_{{\rm P},k}}
e^{jk_0z\cos\theta}
\Bigg\{
\frac{b_k^2k_0}
{w^2(b_k)q_{\rm i}^2(\theta)}
\Big[
A_k^{+}e^{-j\beta(b_k)z}
+
A_k^{-}e^{j\beta(b_k)z}
\Big]
\notag\\[-1mm]
&\times
\Bigg[
\Big[
B_\theta(\theta)\omega\mu_0D(b_k)
-
A_\theta(\theta)\beta(b_k)C(b_k)\cos\theta
\Big]
\notag\\[-1mm]
&\times
\int_{b_k}^{a}
\rho
\Bigg[
\frac{\partial}{\partial\rho}
K_1\!\left(\frac{w(b_k)}{b_k}\rho\right)
\frac{\partial}{\partial\rho}
J_1\!\left(q_{\rm i}(\theta)\rho\right)
\notag\\[-1mm]
&+
\frac{1}{\rho^2}
K_1\!\left(\frac{w(b_k)}{b_k}\rho\right)
J_1\!\left(q_{\rm i}(\theta)\rho\right)
\Bigg]
\,{\rm d}\rho
\notag\\[-1mm]
&+
\Big[
B_\theta(\theta)\beta(b_k)C(b_k)
-
A_\theta(\theta)\omega\mu_0D(b_k)\cos\theta
\Big]
\notag\\[-1mm]
&\times
\int_{b_k}^{a}
\Bigg[
\frac{\partial}{\partial\rho}
K_1\!\left(\frac{w(b_k)}{b_k}\rho\right)
J_1\!\left(q_{\rm i}(\theta)\rho\right)+
K_1\!\left(\frac{w(b_k)}{b_k}\rho\right)
\notag\\[-1mm]
&\times \frac{\partial}{\partial\rho}
J_1\!\left(q_{\rm i}(\theta)\rho\right)
\Bigg]
\,{\rm d}\rho
\Bigg]+
C(b_k)A_\theta(\theta)
\Big[
A_k^{+}e^{-j\beta(b_k)z}
-
A_k^{-}e^{j\beta(b_k)z}
\Big]
\notag\\[-1mm]
&\times
\int_{b_k}^{a}
\rho
K_1\!\left(\frac{w(b_k)}{b_k}\rho\right)
J_1\!\left(q_{\rm i}(\theta)\rho\right)
\,{\rm d}\rho
\Bigg\}
\,{\rm d}z .
\label{eq:PA-proof-theta-radial}
\end{align}
}Moreover, applying integration by parts to the first radial integral
and the product rule to the second, while utilizing
\cite[5.54.1,8.407.1]{Gradshteyn2014} for the remaining
Bessel-product integral together with the order-one Bessel
differential identity $\frac{\partial}{\partial\rho}
\left[
\rho
\frac{\partial}{\partial\rho}
J_1\!\left(q_{\rm i}(\theta)\rho\right)
\right]
=
\left[
\frac{1}{\rho}
-q_{\rm i}^2(\theta)\rho
\right]
J_1\!\left(q_{\rm i}(\theta)\rho\right)$, then \eqref{eq:PA-proof-theta-radial} reduces to
\begin{equation}
\scriptsize
\begin{aligned}
&F_{\theta,k}(\theta,\phi)
=
\frac{\omega^2\mu_0\epsilon_0(1-\epsilon_r)\cos\phi}{4}
\Bigg[
A_k^{+}R_{\theta,k}^{+}(\theta)
\\[-1mm]
&\times
\int_{0}^{L_{{\rm P},k}}
e^{j[k_0\cos\theta-\beta(b_k)]z}\,{\rm d}z
-
A_k^{-}R_{\theta,k}^{-}(\theta)
\\[-1mm]
&\times
\int_{0}^{L_{{\rm P},k}}
e^{j[k_0\cos\theta+\beta(b_k)]z}\,{\rm d}z
\Bigg].
\end{aligned}
\label{eq:PA-proof-theta-radial-reduced}
\end{equation}
where $I_{0,k}$, $I_{1,k}$, and $I_{2,k}$ are given by
\eqref{eq:PA-I0}, \eqref{eq:PA-I1}, and \eqref{eq:PA-I2},
respectively, while $R_{\theta,k}^{+}(\theta)$ and
$R_{\theta,k}^{-}(\theta)$ are the corresponding radial factors
defined in \eqref{eq:PA-radial-theta}. The minus sign associated with
the backward contribution follows from the adopted forward--backward
HE$_{11}$ convention together with the definition of
$R_{\theta,k}^{-}(\theta)$. Finally, the remaining longitudinal
integrations are
\begin{equation}
\small
\begin{aligned}
\int_{0}^{L_{{\rm P},k}}
e^{j[k_0\cos\theta\mp\beta(b_k)]z}\,{\rm d}z
&=
Z_k^\pm(\theta),
\end{aligned}
\label{eq:PA-proof-Z}
\end{equation}
whose substitution into
\eqref{eq:PA-proof-theta-radial-reduced} yields
\eqref{eq:PA-pattern-theta}.

\item[]
\underline{\(\nu=\phi\):}
Similarly, for the $\phi$-polarized reciprocal field, substituting
\eqref{eq:PA-pinched-section-field-closed} and the
$\nu=\phi$ branch of \eqref{eq:PA-reciprocal-total-E2}
into \eqref{eq:PA-reciprocal-radiation-expanded} gives the
azimuthal terms
$\cos\phi'\sin(\phi'-\phi)$ and
$\sin\phi'\cos(\phi'-\phi)$. By utilizing
\cite[2.532.5]{Gradshteyn2014}, their integration over
$0\leq\phi'<2\pi$ gives $-\pi\sin\phi$ and
$\pi\sin\phi$, respectively, which combine with the
corresponding signs of the reciprocal-field components to yield the
common $\sin\phi$ dependence. Following the same
$K_\nu'(x)=-K_{\nu-1}(x)-\nu K_\nu(x)/x$ relation with $\nu=1$,
the same chain-rule manipulations, and the same integration-by-parts,
product-rule, and Bessel-identity steps as for $\nu=\theta$, the
remaining radial terms reduce to the same $I_{0,k}$, $I_{1,k}$, and
$I_{2,k}$ in \eqref{eq:PA-I0}--\eqref{eq:PA-I2}. Collecting the
corresponding terms through $R_{\phi,k}^{+}(\theta)$ and
$R_{\phi,k}^{-}(\theta)$ in \eqref{eq:PA-radial-phi} therefore
gives
\begin{equation}
\scriptsize
\begin{aligned}
&F_{\phi,k}(\theta,\phi)
=
\frac{\omega^2\mu_0\epsilon_0(1-\epsilon_r)\sin\phi}{4}
\Bigg[
A_k^{+}R_{\phi,k}^{+}(\theta)
\\[-1mm]
&\times
\int_{0}^{L_{{\rm P},k}}
e^{j[k_0\cos\theta-\beta(b_k)]z}\,{\rm d}z
-
A_k^{-}R_{\phi,k}^{-}(\theta)
\\[-1mm]
&\times
\int_{0}^{L_{{\rm P},k}}
e^{j[k_0\cos\theta+\beta(b_k)]z}\,{\rm d}z
\Bigg].
\end{aligned}
\label{eq:PA-proof-phi-radial-reduced}
\end{equation}
The remaining longitudinal integrations are the same
$Z_k^{+}(\theta)$ and $Z_k^{-}(\theta)$ obtained in
\eqref{eq:PA-proof-Z}, whose substitution into
\eqref{eq:PA-proof-phi-radial-reduced} yields
\eqref{eq:PA-pattern-phi}.

After deriving the two polarization components, the far-field radiation pattern is obtained by combining the $\theta$- and $\phi$-polarized components in the spherical basis as
\begin{equation}
\small
\mathbf F_k(\theta,\phi)
=
F_{\theta,k}(\theta,\phi)
\widehat{\boldsymbol{\theta}}
+
F_{\phi,k}(\theta,\phi)
\widehat{\boldsymbol{\phi}}.
\label{eq:PA-pattern-vector-proof}
\end{equation}
Substituting \eqref{eq:PA-pattern-theta} and
\eqref{eq:PA-pattern-phi} into this expression and restoring the common
far-field dependence $e^{-jk_0R}/R$ gives
\eqref{eq:PA-far-field}, which concludes the proof.
\end{itemize}

\bibliographystyle{IEEEtran}
\bibliography{Reference}

\end{document}